\documentclass[a4paper,english]{lipics-v2018}

\newif\ifarxiv
\arxivtrue

\usepackage{microtype}
\usepackage{amsmath}
\usepackage{amssymb}
\usepackage{amsthm}
\usepackage{stmaryrd} 
\usepackage{algorithm}
\usepackage{algpseudocode} 
\usepackage{setspace}
\usepackage{csquotes}
\usepackage{tikz}
\usepackage{mdwlist}
\usepackage{setspace}
\usepackage{calc}
\usepackage{paralist}

\usetikzlibrary{arrows.meta}
\usepackage{isabelle,isabellesym}
\isabellestyle{it}
\renewenvironment{isabelle}{%
  \bigbreak\noindent\hspace{\parindent}%
  \begin{minipage}{\textwidth-\parindent}
  \begin{isabellebody}%
  \begin{tabbing}%
}{%
  \end{tabbing}%
  \end{isabellebody}%
  \end{minipage}%
  \bigbreak%
}
\renewcommand{\isanewline}{\\}

\ifarxiv
  \title{OpSets: Sequential Specifications for Replicated Datatypes (Extended Version)}
  \titlerunning{OpSets: Sequential Specifications for Replicated Datatypes}
\else
  \title{OpSets: Sequential Specifications for Replicated Datatypes}
\fi

\author{Martin Kleppmann}{Computer Laboratory, University of Cambridge, UK}{mk428@cl.cam.ac.uk}{https://orcid.org/0000-0001-7252-6958}{}
\author{Victor B.\ F.\ Gomes}{Computer Laboratory, University of Cambridge, UK}{vb358@cl.cam.ac.uk}{https://orcid.org/0000-0002-2954-4648}{}
\author{Dominic P.\ Mulligan}{Security Research Group, Arm Research, Cambridge, UK}{Dominic.Mulligan@arm.com}{https://orcid.org/0000-0003-4643-3541}{}
\author{Alastair R.\ Beresford}{Computer Laboratory, University of Cambridge, UK}{arb33@cl.cam.ac.uk}{https://orcid.org/0000-0003-0818-6535}{}
\authorrunning{M.\ Kleppmann, V.\,B.\,F.\ Gomes, D.\,P.\ Mulligan, and A.\,R.\ Beresford} 

\Copyright{Martin Kleppmann, Victor B.\ F.\ Gomes, Dominic P.\ Mulligan, and Alastair R.\ Beresford}

\subjclass{%
;
;
;
;
;
}


\ifarxiv
  \hideLIPIcs
  \nolinenumbers
\else
  \relatedversion{Extended paper with proofs: \url{http://martinkl.com/papers/opsets.pdf}}
\fi

\EventEditors{}
\EventNoEds{0}
\EventLongTitle{}
\EventShortTitle{}
\EventAcronym{}
\EventYear{}
\EventDate{}
\EventLocation{}
\EventLogo{}
\SeriesVolume{}
\ArticleNo{~}

\begin{document}
\maketitle

\begin{abstract}
We introduce OpSets, an executable framework for specifying and reasoning about the semantics of replicated datatypes that provide eventual consistency in a distributed system, and for mechanically verifying algorithms that implement these datatypes.
Our approach is simple but expressive, allowing us to succinctly specify a variety of abstract datatypes, including maps, sets, lists, text, graphs, trees, and registers.
Our datatypes are also composable, enabling the construction of complex data structures.
To demonstrate the utility of OpSets for analysing replication algorithms, we highlight an important correctness property for collaborative text editing that has traditionally been overlooked; algorithms that do not satisfy this property can exhibit awkward interleaving of text.
We use OpSets to specify this correctness property and prove that although one existing replication algorithm satisfies this property, several other published algorithms do not.
We also show how OpSets can be used to develop new replicated datatypes: we provide a simple specification of an atomic move operation for trees, an operation that had previously been thought to be impossible to implement without locking.
We use the Isabelle/HOL proof assistant to formalise the OpSets approach and produce mechanised proofs of correctness of the main claims in this paper, thereby eliminating the ambiguity of previous informal approaches, and ruling out reasoning errors that could occur in handwritten proofs.
\end{abstract}
\clearpage

\section{Introduction}

A common requirement across many distributed systems is that several nodes may concurrently access and manipulate some shared data structure.
Examples include everything from journalists using their laptops to work on a shared text document to a set of web servers manipulating a common database.
In doing so, it is important that the shared data satisfies certain \emph{consistency guarantees}.
For example, \emph{strong consistency models} such as serializability \cite{Kleppmann:2017wj} or linearizability \cite{Herlihy:1990jq} make a system behave like a single sequentially executing node, even when it is in fact replicated and concurrent.
An unavoidable downside of these models is that any operation or transaction must wait for network communication before it is allowed to complete \cite{Davidson:1985hv,Gilbert:2002il}.
Thus, in a system with strong consistency, a node cannot make progress while it is offline or partitioned from other nodes.

On the other hand, \emph{eventual consistency} \cite{Bailis:2013jc,Burckhardt:2014hy,Terry:1994fp,Vogels:2009ca} allows each participant to modify a local copy (\emph{replica}) of a shared data structure while offline, but its definition is very weak: \emph{``if no new updates are made to the shared state, all nodes will eventually have the same data.''}
The premise \emph{``if no new updates are made''} may never be true if the shared state is continually modified (i.e.\ the system is never quiescent).
Moreover, nothing in the definition of eventual consistency specifies which final states are legal.

Conflict-free Replicated Data Types, or CRDTs \cite{Shapiro:2011wy,Shapiro:2011un}, are abstractions for replicated state that have received significant attention in recent years (see \S~\ref{sec:relwork}).
The primary correctness property for CRDTs is \emph{convergence} \cite{Shapiro:2011un,Gomes:2017gy}, defined as: \emph{``whenever any two replicas have applied the same set of updates, they are in the same state''}, even if each replica applies the updates in a different order.
Convergence is a stronger property than eventual consistency, but it also fails to define what exactly the converged state should be.

In this work we introduce \emph{Operation Sets} (or \emph{OpSets} for short), a novel approach for specifying the semantics of replicated datatypes, and for reasoning about algorithms for concurrent data access and manipulation.
We go beyond merely ensuring replica convergence: the OpSets approach is an executable specification that precisely defines the permitted states of a replica after some set of updates has been applied.
Our contributions in this paper are as follows:

\begin{itemize*}
\item In \S~\ref{sec:approach} we introduce the OpSet, which provides a simple abstraction for specifying and reasoning about the consistency properties of concurrently editable data structures.

\item On top of this abstraction, in \S~\ref{sec:datatypes} and \S~\ref{sec:tree}, we specify a variety of composable abstract datatypes (maps, sets, lists, text, graphs, trees, and registers), and we argue that our specifications are both simple and precise, making them a suitable tool for reasoning about replicated data.

\item In \S~\ref{sec:bad-merge} we demonstrate how the OpSet abstraction can be used to reason about existing algorithms.
We highlight an important correctness property for collaborative text editing that has been overlooked by prior work in this area.
Our specification is, to our knowledge, the first that correctly captures this property.
We then review a selection of text editing CRDTs from the literature, prove that one satisfies our specification, and identify several others that fail to satisfy our correctness property.

\item In \S~\ref{sec:tree} we show how the OpSet abstraction can be used to develop new replicated datatypes. In particular we describe, for the first time, how an atomic move operation can be defined for a tree CRDT.
This operation can be used to move a subtree to a new position within the tree, or to rename a key in a map, or to reorder items in a list.
The OpSets approach enables a simple definition of this operation that had previously been thought impossible to implement without locking \cite{Najafzadeh:2017vk,Najafzadeh:2018bw}.

\item Using the Isabelle/HOL proof assistant~\cite{DBLP:conf/tphol/WenzelPN08} we formalise the OpSets approach, producing mechanised proofs of correctness of the main claims in this paper.
In particular, we prove that our list specification is strictly stronger than the recent specification of collaborative text editing by Attiya et al. \cite{Attiya:2016kh}.
By using mechanised proofs we eliminate the ambiguity of previous informal approaches, and rule out reasoning errors that could occur in handwritten proofs.
Moreover, the proof framework we have developed is reusable and can be leveraged to verify other datatypes in the future.
\end{itemize*}

\ifarxiv
  The appendices contain
\fi
an overview of our Isabelle/HOL mechanisation, and pseudocode for the replicated datatypes discussed in this paper.
The full formal proof development is published in the Isabelle Archive of Formal Proofs \cite{AFP}.

\section{The OpSets Approach}\label{sec:approach}

The OpSets approach is a simple abstraction for describing the consistency properties of a replicated data system.
We outline the general approach in this section, before describing concrete data structures and specifications in \S~\ref{sec:datatypes} and \S~\ref{sec:tree}.

\subsection{System Model}\label{sec:system-model}

We assume that the system consists of a set of \emph{nodes} connected by a \emph{network}.
These nodes concurrently access some \emph{shared data structure}, which may be a relational database (consisting of rows in tables), a text document (a sequence of characters), a vector-graphics document (a tree of records describing graphical objects), a filesystem (a tree of directories and files), or any other kind of data structure.

New nodes can be added at any time, and the set of nodes need not be known in advance.
Nodes might be mobile devices, and hence we assume that nodes are sometimes \emph{offline}, i.e.\ temporarily unable to communicate with other nodes.
We require that nodes can access the shared data anytime, even while offline.
Thus, each node has a local copy of the shared data structure, which it can read and modify without waiting for any communication or coordination with other nodes.

Whenever a node makes a modification to that structure, it records the change as an \emph{operation}.
For example, an operation may describe an insertion at a particular position in a text document.
Each node locally maintains a set of operations, the \emph{OpSet}.
Whenever a node makes a change to the shared data, it adds the corresponding operation to its OpSet, and also sends \emph{messages} containing the operation to other nodes.
Whenever a node receives a message from another node, the operation in that message is added to the recipient's local OpSet.
Operations remain immutable throughout this process.

We make no assumptions about the reliability of the network: messages may be lost, duplicated, or arbitrarily reordered.
Reflecting the characteristics of real networks, we assume that lost messages are retransmitted when possible (e.g.\ using TCP), but messages may be permanently lost due to network or node failures.
Since the OpSet at each node is a monotonically growing set of operations, any two communicating nodes can merge their OpSets using the standard set union operator $\cup$.
Set union is commutative, associative, and idempotent, ensuring that communicating nodes converge towards the same OpSet contents.

We assume that each operation has a unique identifier (ID), that new IDs can be generated by any node without communication with other nodes, and that we have a total ordering on operation IDs.
These requirements can easily be met by using Lamport timestamps \cite{Lamport:1978jq} as IDs.
A Lamport timestamp is a pair $(\mathit{counter}, \mathit{nodeID})$ that is constructed as follows:
\begin{itemize}
\item $\mathit{counter}$ is an integer.
    To generate a new ID, find the maximum counter of any existing operation ID in the local OpSet, and increment that number.
\item $\mathit{nodeID}$ is a string that uniquely identifies the node generating the ID, e.g.\ a UUID \cite{Leach:2005hm}.
\end{itemize}

Although different nodes may generate IDs with the same counter value, each node generates IDs with strictly monotonically increasing counter values, and thus IDs are globally unique.
We define the total order on IDs as being the lexicographic order:
\[
    (\mathit{ctr}_1, \mathit{node}_1) < (\mathit{ctr}_2, \mathit{node}_2)
    \;\Longleftrightarrow\;
    \mathit{ctr}_1 < \mathit{ctr}_2 \;\vee\;
    (\mathit{ctr}_1 = \mathit{ctr}_2 \wedge \mathit{node}_1 <\mathit{node}_2).
\]


\subsection{Interpreting an OpSet}\label{sec:op-serial}

Most definitions of operation-based CRDTs describe how a node's local state is manipulated by operations \cite{Shapiro:2011wy,Shapiro:2011un}.
We now depart from this convention and present an alternative formulation of replicated datatypes.

In the OpSets approach, we require that the shared data structure is never manipulated directly.
Instead, we use an \emph{interpretation function} $\llbracket-\rrbracket$ that takes an OpSet $O$ and returns the current state $\llbracket O \rrbracket$ of the shared data structure described by the OpSet.
The interpretation function is \emph{pure}, i.e.\ deterministic, side-effect free, and its result depends only on $O$.
All nodes in the system employ the same interpretation function.

Consequently, whenever any two nodes have the same OpSet $O$, their view of the shared data structure $\llbracket O \rrbracket$ must also be equal.
This construction trivially ensures eventual consistency: as two nodes converge towards the same OpSet contents, any data structure that is deterministically derived from the OpSet must also converge.

In principle, any deterministic function can serve as interpretation function.
However, in defining the semantics of CRDTs (see \S~\ref{sec:datatypes} and \S~\ref{sec:tree}), we have found it useful to specialise $\llbracket-\rrbracket$ such that we can interpret one operation at a time.

Let the OpSet $O$ be a set of pairs $(\mathit{id},\, \mathit{op})$, where $\mathit{id}$ is a unique operation identifier and $\mathit{op}$ is an arbitrary description of the change that occurred.
Assume that we have a total ordering $<$ on identifiers, as explained in \S~\ref{sec:system-model}.
Then observe that for any OpSet there exists a unique sequence of operations, containing all operations of the OpSet in ascending order of their identifier.
We can specify the semantics of each operation~--- that is, the effect of the operation on the OpSet interpretation~--- when applied in this sequential order.

Formally, we can define the interpretation $\llbracket O \rrbracket$ of the OpSet $O$ as follows:
\begin{align*}
    \big\llbracket \emptyset \big\rrbracket &= \mathsf{InitialState} \\
    \big\llbracket O \;\cup\; \{(\mathit{id},\, \mathit{op})\} \big\rrbracket &=
    \mathsf{interp}\big[\llbracket O \rrbracket,\, (\mathit{id},\, \mathit{op})\big]
    \qquad\text{ provided that } \forall\,(\mathit{id}',\, \mathit{op}') \in O.\; \mathit{id}' < \mathit{id}
\end{align*}
where $\mathsf{interp}\big[S,\, (\mathit{id},\, \mathit{op})\big]$ is the interpretation of the operation $(\mathit{id},\, \mathit{op})$ in the state $S$, and $\mathsf{InitialState}$ is a fixed minimal element (e.g. the empty tree, or empty list) of the replicated type described.
In other words, if $S$ is the result of interpreting all operations with identifiers less than $\mathit{id}$, then
$\mathsf{interp}\big[S,\, (\mathit{id},\, \mathit{op})\big]$ is the interpretation of the OpSet to which $(\mathit{id},\, \mathit{op})$ has been added.
For example, if $\mathit{id}_1 < \mathit{id}_2 < \mathit{id}_3$, we have:
\begin{align*}
    \big\llbracket \{(\mathit{id}_1,\ \mathit{op}_1),\;
    &(\mathit{id}_2,\ \mathit{op}_2),\,
    (\mathit{id}_3,\ \mathit{op}_3)\} \big\rrbracket \;=\\
    &\mathsf{interp}\big[\mathsf{interp}\big[\mathsf{interp}\big[\mathsf{InitialState},\,
    (\mathit{id}_1,\ \mathit{op}_1)\big],\,
    (\mathit{id}_2,\ \mathit{op}_2)\big],\,
    (\mathit{id}_3,\ \mathit{op}_3)\big]
\end{align*}
Provided that the operation interpretation $\mathsf{interp}\big[S,\, (\mathit{id},\, \mathit{op})\big]$ is deterministic, the OpSet interpretation function $\llbracket-\rrbracket$ is also deterministic, due to the fact that the operation order in the OpSet is unique.

\subsection{Receiving Messages Out-of-order}\label{sec:order-change}

Many computing systems are based on the idea of putting operations in some total order, and executing them in that order.
For example, serializable transactions \cite{Kleppmann:2017wj} and state machine replication \cite{Schneider:1990vy} follow this approach.
However, it is important to understand that the OpSet interpretation of \S~\ref{sec:op-serial} relies on a weaker notion of ordering than most systems.

With serializable transactions and state machine replication, once a transaction/operation has been executed in some state, its results are expected to be durable.
Thus, before executing some transaction $T_i$, the system needs to ensure that there is no pending transaction with a lower ID than $T_i$ (which would need to be executed before $T_i$), since otherwise the subsequent arrival of a transaction with lower ID would invalidate the state in which $T_i$ was executed.
However, ensuring this precondition is expensive: as we show in \S~\ref{sec:op-sequences}, it requires communication with at least a quorum of nodes; if the IDs are Lamport timestamps, it even requires communication with every single node \cite{Lamport:1978jq}.
If too many nodes are offline, the system cannot execute any transactions.

By contrast, our system model of \S~\ref{sec:system-model} requires nodes to always be able to read and modify the shared data, even when all nodes are offline.
Moreover, we do not assume any ordering guarantees from the network.
Thus, whenever there is some operation $(\mathit{id}_1, \mathit{op}_1) \in O$ in the OpSet $O$ of some node, it is possible that the node will subsequently receive a message containing $(\mathit{id}_2, \mathit{op}_2)$, where $\mathit{id}_2 < \mathit{id}_1$; that is, the later-arriving operation needs to be applied before the existing operation $(\mathit{id}_1, \mathit{op}_1)$ in the OpSet interpretation $\llbracket O \rrbracket$.

In the OpSet model, such out-of-order delivery of operations is no problem: the order in which operations are received has no effect on the OpSet $O$, and since we assume the interpretation function to be pure and side-effect free, the interpretation $\llbracket O \rrbracket$ can always be recomputed whenever new operations are added to $O$.

The interpretation function is an \emph{executable specification} that defines the expected result of interpreting a set of operations.
Presenting replicated datatypes in this manner has two significant advantages:
\begin{enumerate*}
\item
Unlike typical definitions of CRDT algorithms \cite{Shapiro:2011wy,Shapiro:2011un}, it is not necessary for the interpretation function $\mathsf{interp}\big[S,\, (\mathit{id},\, \mathit{op})\big]$ to commute with respect to other operations: any pure function can be used.
This fact makes it much simpler to specify the interpretation of operations, as we shall see in \S~\ref{sec:datatypes} and \S~\ref{sec:tree}.
\item
We can guarantee the existence of an implementation of each described datatype: the specification itself.
This is in contrast to axiomatic specifications, which may not be implementable, and require additional work to demonstrate than an implementation exists which satisfies the axiomatic description.
\end{enumerate*}

For practical implementations of replicated datatypes, a naive OpSet interpretation may exhibit poor performance, since nodes must potentially apply the same subset of operations repeatedly.
More efficient (and, most likely, more complex) algorithms for CRDTs can therefore be developed and shown to satisfy the OpSet-based specification---we do this in \S~\ref{sec:bad-merge}.

However, we have developed a practical JavaScript CRDT implementation around the OpSet model \cite{Automerge}, and found it to have some advantages: for example, users can easily inspect the editing history of a document, since every past version of the document is the interpretation of a particular subset of operations.
Moreover, using OpSets provides a straightforward mechanism for recovering from network partitions and failures, as missing operations may be retransmitted and added to the OpSets of previously partitioned nodes.
The details of this implementation are beyond the scope of this paper.

\section{Specifying a Graph of Lists, Maps, and Registers}\label{sec:datatypes}

We now make the OpSets approach concrete by defining example semantics for commonly-used data structures: maps (which associate values with user-specified keys) and lists (linear sequences of values).
The map datatype can also represent a set (by using keys as members of the set, and ignoring values).
The list datatype can also represent text (by mapping each character to a list element).
In both lists and maps the values may be primitives (such as numbers or strings), or references to other map or list objects.
Using these references we can construct arbitrary object graphs, including cycles of object references, like in object-oriented programming languages.
In \S~\ref{sec:tree} we will show how to restrict this object graph so that it conforms to a tree structure.

We treat each key of a map, and each element of a list, as a multi-value register.
That is, if there are several concurrent assignments to the same map key or list element, our datatype preserves all concurrently written values.
Thus, reading a map key or list element may return multiple values, which may be merged explicitly by the user.
Assigning a new value to a map key or list element overwrites all causally preceding values.
Different register behaviour, such as last-writer-wins (arbitrarily picking one of the concurrently written values as winner), can easily be defined, as we show later.

\subsection{Generating Operations}\label{sec:datatypes-gen}

An OpSet for these datatypes may contain six types of operation:
\begin{itemize}
    \item $(\mathit{id},\, \mathsf{MakeMap})$ creates a new, empty map object that is identified by $\mathit{id}$.
    \item $(\mathit{id},\, \mathsf{MakeList})$ creates a new, empty list object that is identified by $\mathit{id}$.
    \item $(\mathit{id},\, \mathsf{MakeVal}(\mathit{val}))$ associates the ID $\mathit{id}$ with the primitive value $\mathit{val}$ (e.g.\ a number, string, or boolean).
        This operation is used to ``wrap'' any primitive value, allowing $\mathsf{Assign}$ operations (see below) to always use IDs as values, regardless of whether the value is a primitive value, or a reference to a map or list object.
    \item $(\mathit{id},\, \mathsf{InsertAfter}(\mathit{ref}))$ creates a new list element with ID $\mathit{id}$, and inserts it into a list.
        If $\mathit{ref}$ is the ID of a prior $\mathsf{MakeList}$ operation, then the new element is inserted at the head of that list.
        Otherwise $\mathit{ref}$ must be the ID of an existing list element (i.e.\ a prior $\mathsf{InsertAfter}$ operation), in which case the new list element is inserted immediately after the referenced list element.
        Note that the $\mathsf{InsertAfter}$ operation does not associate a value with the new list element; that is done by a subsequent $\mathsf{Assign}$ operation.
    \item $(\mathit{id},\, \mathsf{Assign}(\mathit{obj}, \mathit{key}, \mathit{val}, \mathit{prev}))$ assigns a new value to a key within a map (if $\mathit{obj}$ is the ID of a prior $\mathsf{MakeMap}$ operation), or to a list element (if $\mathit{obj}$ is the ID of a prior $\mathsf{MakeList}$ operation).
        In the case of map assignment, $\mathit{key}$ is the user-specified key to be updated, which may be any primitive value such as a string or integer.
        In the case of a list, $\mathit{key}$ is the ID of the list element to be updated (i.e.\ the ID of a prior $\mathsf{InsertAfter}$ operation).
        $\mathit{val}$ is the ID of the value being assigned, which may identify a $\mathsf{MakeMap}$, $\mathsf{MakeList}$, or $\mathsf{MakeVal}$ operation.
        $\mathit{prev}$ is the set of IDs of prior $\mathsf{Assign}$ operations to the same key in the same object, which are overwritten by the present operation.
    \item $(\mathit{id},\, \mathsf{Remove}(\mathit{obj}, \mathit{key}, \mathit{prev}))$ removes a key-value pair from a map, or an element from a list.
        As with $\mathsf{Assign}$, $\mathit{obj}$ is the ID of the prior $\mathsf{MakeMap}$ or $\mathsf{MakeList}$ operation that created the object being updated, and $\mathit{key}$ identifies the key or list element being removed.
        $\mathit{prev}$ is the set of IDs of prior $\mathsf{Assign}$ operations to the same key in the same object, which are removed by the present operation.
\end{itemize}

\ifarxiv
  \noindent
  Pseudocode for generating these operations is given in Appendix~\ref{sect:appendix:generating-ops}.
\fi

\subsection{Interpreting Operations}\label{sec:datatypes-interp}

We use the sequential OpSet interpretation given in \S~\ref{sec:op-serial}.
To encode the current state of map and list data structures we use a pair of relations $(E,\, L)$:

\begin{description}
    \item[The element relation $E \subseteq (\mathrm{ID} \times \mathrm{ID} \times (\mathrm{ID} \cup \mathrm{Key}) \times \mathrm{ID})$]
        is a set of 4-tuples containing the values currently assigned to map keys and list elements.
        If $(\mathit{id}, \mathit{obj}, \mathit{key}, \mathit{val}) \in E$, then an $\mathsf{Assign}$ operation with ID $\mathit{id}$ updated the object with ID $\mathit{obj}$, assigning the value with ID $\mathit{val}$ to the map key or list element $\mathit{key}$.
        If $\mathit{obj}$ references a list object, $\mathit{key}$ is the ID of an element in the list relation $L$ (see below).
        If $\mathit{obj}$ references a map object, any primitive value such as string or integer may be used as $\mathit{key}$.
    \item[The list relation $L \subseteq (\mathrm{ID} \times (\mathrm{ID} \cup \{\bot\}))$] is a set of pairs that indicates the order of list elements.
        If $(\mathit{prev}, \mathit{next}) \in L$, that means the list element with ID $\mathit{prev}$ is immediately followed by the list element with ID $\mathit{next}$.
        We use $(\mathit{last}, \bot) \in L$ to indicate that list element $\mathit{last}$ has no successor.
        To indicate that $\mathit{head}$ is the first element in the list $\mathit{obj}$ (i.e.\ $\mathit{obj}$ is the ID of the $\mathsf{MakeList}$ operation that created the list) we have $(\mathit{obj}, \mathit{head}) \in L$.
\end{description}

\noindent
Initially, both relations are empty; that is, we have $\llbracket\emptyset\rrbracket = \mathsf{InitialState} = (\emptyset,\, \emptyset)$.
We can then define the interpretation of the six operation types as follows:
\begin{align*}
    \mathsf{interp}\big[(E,\, L),\; &(\mathit{id},\, \mathsf{Assign}(\mathit{obj}, \mathit{key}, \mathit{val}, \mathit{prev})) \big] \;=\\
    &\Big( \big\{ (\mathit{id}', \mathit{obj}', \mathit{key}', \mathit{val}') \in E \mid
    \mathit{id}' \notin \mathit{prev} \big\} \;\cup\;
    \big\{ (\mathit{id}, \mathit{obj}, \mathit{key}, \mathit{val}) \big\},\; L \Big) \\[5pt]
    \mathsf{interp}\big[(E,\, L),\; &(\mathit{id},\, \mathsf{Remove}(\mathit{obj}, \mathit{key}, \mathit{prev})) \big] \;=\\
    &\Big( \big\{ (\mathit{id}', \mathit{obj}', \mathit{key}', \mathit{val}') \in E \mid
    \mathit{id}' \notin \mathit{prev} \big\},\; L \Big) \\[5pt]
    \mathsf{interp}\big[(E,\, L),\; &(\mathit{id},\, \mathsf{InsertAfter}(\mathit{ref})) \big] \;=\\
    &\left\{
        \arraycolsep=0pt \def\arraystretch{1.2}
        \begin{array}{ll}
            (E,\, L) & \text{if } \nexists\,n.\; (\mathit{ref},\, n) \in L\\[3pt]
            \Big(E,\; \big\{ (p,n) &\;\in L \mid p \neq \mathit{ref} \big\} \;\cup\;
            \big\{ (\mathit{ref}, \mathit{id}) \big\} \;\cup\;
            \big\{ (\mathit{id}, n) \mid (\mathit{ref}, n) \in L \big\} \Big)\\
            & \text{if } \exists\,n.\; (\mathit{ref},\,n) \in L
        \end{array} \right. \\[5pt]
    \mathsf{interp}\big[(E,\, L),\; &(\mathit{id},\, \mathsf{MakeList}) \big] \hspace{21.1pt}=\;
    \Big( E,\; L \;\cup\; \big\{(\mathit{id},\, \bot)\big\}\Big) \\[1pt]
    \mathsf{interp}\big[(E,\, L),\; &(\mathit{id},\, \mathsf{MakeMap}) \big] \hspace{17.6pt}=\; (E,\; L) \\[5pt]
    \mathsf{interp}\big[(E,\, L),\; &(\mathit{id},\, \mathsf{MakeVal}(\mathit{val})) \big] \;=\; (E,\; L)
\end{align*}

The interpretation of $\mathsf{Assign}$ and $\mathsf{Remove}$ updates only $E$ and leaves $L$ unchanged; conversely, the interpretation of $\mathsf{InsertAfter}$ and $\mathsf{MakeList}$ updates only $L$.
Both the $\mathsf{Assign}$ and $\mathsf{Remove}$ interpretations remove any tuples from causally prior assignments (those whose IDs appear in $\mathit{prev}$), but leave any tuples from concurrent assignments unchanged.
This is the behaviour of a multi-value register; if a last-writer-wins register is required, the condition $\mathit{id}' \notin \mathit{prev}$ can be changed to $\mathit{obj}' \neq \mathit{obj} \;\vee\; \mathit{key}' \neq \mathit{key}$, which removes any existing tuples with the same object ID and key.

The interpretation of $\mathsf{InsertAfter}$ resembles the insertion into a linked list, as illustrated in Figure~\ref{fig:list-insert}.
For example, to interpret $(\mathit{id},\, \mathsf{InsertAfter}(\mathit{ref}))$, if we have $(\mathit{ref}, \mathit{next}) \in L$, we remove the pair $(\mathit{ref}, \mathit{next})$ from L, and add the pairs $(\mathit{ref}, \mathit{id})$ and $(\mathit{id}, \mathit{next})$ to $L$.
Thus, the new list element $\mathit{id}$ is inserted between $\mathit{ref}$ and $\mathit{next}$.

\begin{figure}
\centering
\begin{tikzpicture}
  \tikzstyle{every node}=[anchor=base,minimum width=6mm,text height=7pt,text depth=2pt,font=\footnotesize]
  \node [anchor=west] at (-7,1.7) {Before:};
  \node [anchor=west] at (-7,0.3) {$L = \{ (2, 13),\; (13, 5),\; (5, 23),\; (23, \bot) \}$};
  \node [anchor=west] at (0,1.7) {After adding $(25,\, \mathsf{InsertAfter}(13))$ to OpSet:};
  \node [anchor=west] at (0,-0.7) {$L' = L \;-\; \{(13, 5)\} \;\cup\; \{(13, 25),\; (25, 5)\}$};
    \matrix [column sep={6mm,between origins},nodes=draw,matrix anchor=west] at (-7,1) {
    \node (l1a) {2};  & \node (l1b) {13}; &&
    \node (l2a) {13}; & \node (l2b) {5};  &&
    \node (l3a) {5};  & \node (l3b) {23}; &&
    \node (l4a) [draw=none] {None}; \\
  };
  \draw [->] (l1b) -- (l2a);
  \draw [->] (l2b) -- (l3a);
  \draw [->] (l3b) -- (l4a);
  \matrix [column sep={6mm,between origins},nodes=draw,matrix anchor=west] at (1,0) {
    \node (n1) {13}; & \node (n2) {25}; &&
    \node (n3) {25}; & \node (n4) {5}; \\
  };
  \matrix [column sep={6mm,between origins},nodes=draw,matrix anchor=west] at (0,1) {
    \node (r1a) {2};  & \node (r1b) {13}; &&&&&
    \node (r3a) {5};  & \node (r3b) {23}; &&
    \node (r4a) [draw=none] {None}; \\
  };
  \draw [->] (r1b.east) .. controls (2.7,1) and (-0.3,0) .. (n1.west);
  \draw [->] (n2) -- (n3);
  \draw [->] (n4.east) .. controls (5.6,0) and (2.3,1) .. (r3a.west);
  \draw [->] (r3b) -- (r4a);
\end{tikzpicture}
\caption{Illustration of the interpretation of an $\mathsf{InsertAfter}$ operation.}\label{fig:list-insert}
\end{figure}

Note that $L$ never shrinks, it only ever grows through interpreting $\mathsf{InsertAfter}$ operations.
When a list element is removed by a $\mathsf{Remove}$ operation, the effect is that all values are removed from the list element in the element relation $E$, but the list element remains in $L$ as a \emph{tombstone}, so that any concurrent $\mathsf{InsertAfter}$ operations can still locate the referenced list position.
Thus, from a user's point of view a list element only exists if it has at least one associated value in the $E$ relation; any list elements without an associated value should be ignored.

\section{Discussion: Merging Text Edits}\label{sec:bad-merge}

The datatypes we have specified in \S~\ref{sec:datatypes} can support a wide range of applications.
For example, the list datatype can be used to implement a collaborative text editor: by treating the text as a list of individual characters, every edit can be expressed as a sequence of insertion or deletion operations on the list.

The problem of collaborative text editing has been studied extensively, using two main approaches: Operational Transformation and CRDTs.
We discuss this prior work in \S~\ref{sec:relwork}.
We will now highlight a scenario that, to our knowledge, has not been considered by any previous work on collaborative text editing.

Consider the execution illustrated in Figure~\ref{fig:bad-merge}.
In this example, two users are concurrently editing a text document that initially reads ``Hello!''.
The user on the left changes it to read ``Hello Alice!'', while concurrently the user on the right changes the document to read ``Hello Charlie!''.
When the concurrent edits are merged, the algorithm randomly interleaves the two insertions of ``~Alice'' and ``~Charlie'' character by character, resulting in an unreadable jumble of characters.

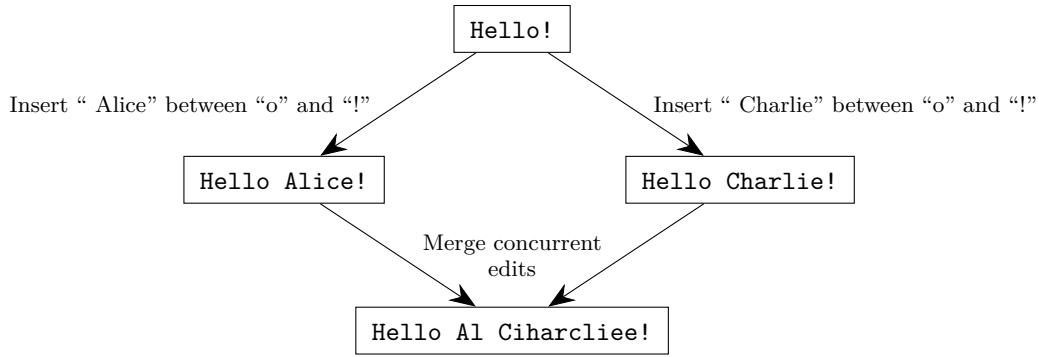
\begin{figure}
\centering
\begin{tikzpicture}
  \tikzstyle{box}=[rectangle,draw,inner xsep=6pt,text height=9pt,text depth=2pt]
  \tikzstyle{every path}=[draw,-{Stealth[length=3.5mm]}]
  \node [box] (start) at (3,4) {\texttt{Hello!}};
  \node [box] (left)  at (0,2) {\texttt{Hello Alice!}};
  \node [box] (right) at (6,2) {\texttt{Hello Charlie!}};
  \node [box] (merge) at (3,0) {\texttt{Hello Al Ciharcliee!}};
  \draw (start) to node [left,inner xsep=10pt,font=\footnotesize]  {Insert ``~Alice'' between ``o'' and ``!''} (left);
  \draw (start) to node [right,inner xsep=10pt,font=\footnotesize] {Insert ``~Charlie'' between ``o'' and ``!''} (right);
  \draw (left)  -- (merge);
  \draw (right) -- (merge);
  \node [text width=3cm,text badly centered,font=\footnotesize] at (3,1) {Merge concurrent edits};
\end{tikzpicture}
\caption{Two concurrent insertions at the same position are interleaved.}\label{fig:bad-merge}
\end{figure}

The problem is even worse if the concurrent insertions are not just a single word, but an entire paragraph or section.
In these cases, interleaving the users' insertions would most likely result in an entirely incomprehensible text that would have to be deleted and rewritten.
Even though the merge in Figure~\ref{fig:bad-merge} is so obviously undesirable, there is to our knowledge no formal specification of collaborative text editing that rules out such an interleaving of insertions.

\begin{theorem}\label{thm:attiya-allows-interleaving}
    The $\mathcal{A}_\textsf{strong}$ specification of collaborative text editing by Attiya et al. \cite{Attiya:2016kh} allows the outcome in Figure~\ref{fig:bad-merge}; that is, an algorithm that interleaves concurrent insertions at the same position may nevertheless satisfy the $\mathcal{A}_\textsf{strong}$ specification.
    Moreover, the text editing CRDT algorithms Logoot \cite{Weiss:2009ht,Weiss:2010hx} and LSEQ \cite{Nedelec:2016eo,Nedelec:2013ky} also allow this outcome.
\end{theorem}
\begin{proof}
    Follows directly from the respective definitions, which are all based on the idea of assigning each character a position in a totally ordered identifier space, such that the order of identifiers corresponds to the order of characters in the document.
    When a new character is inserted, it is assigned an identifier that lies somewhere between the identifiers of its predecessor and successor.
    However, when concurrent insertions with the same predecessor and successor are performed, those insertions are ordered arbitrarily.
    Repeated insertions within the same predecessor-successor interval may thus be interleaved arbitrarily.

    We also performed tests with open source implementations of Logoot \cite{AhmedNacer:2011ke,ReplicationBenchmark} and LSEQ \cite{LSEQTree,Nedelec:2016eo}, and observed this interleaving anomaly occurring in practice.
\end{proof}

Rather than interleaving characters, a better approach to merging is to keep all insertions by a particular user together as a continuous sequence.
With this constraint, there are two acceptable merged results in the example of Figure~\ref{fig:bad-merge}: either ``Hello Alice Charlie!'' or ``Hello Charlie Alice!''.
The choice between these two outcomes is arbitrary, as there is no \emph{a priori} requirement for one user's insertions to come before the other's.

\begin{theorem}\label{thm:no-interleaving}
    The list specification from \S~\ref{sec:datatypes} does not allow interleaving of concurrent insertions.
    That is, if one user inserts a character sequence $\langle x_1, x_2, \dots, x_n \rangle$ and another user concurrently inserts a character sequence $\langle y_1, y_2, \dots, y_m \rangle$ at the same position, the merged document contains either the character sequence $\langle x_1, x_2, \dots, x_n, y_1, y_2, \dots, y_m \rangle$ or the character sequence $\langle y_1, y_2, \dots, y_m, x_1, x_2, \dots, x_n \rangle$ at the specified position.
\end{theorem}
\begin{proof}
    We formalise the list specification and Theorem~\ref{thm:no-interleaving} using the Isabelle/HOL proof assistant~\cite{DBLP:conf/tphol/WenzelPN08}.
    \ifarxiv
        The formal proof development is summarised in Appendix~\ref{appendix:no-interleaving}.
    \fi
\end{proof}

For an informal argument why interleaving is ruled out, see Figure~\ref{fig:op-permutations}, which shows an editing scenario similar to Figure~\ref{fig:bad-merge}, but with the insertions of ``~Alice'' and ``~Charlie'' shortened to ``Al'' and ``Ch'' respectively.
The example contains four insertion operations (``A'', ``l'', ``C'', and ``h''), which can be ordered in six possible ways.
However, among the six possible operation orderings there are only two possible results: \texttt{ChAl} or \texttt{AlCh}.
Interleavings such as \texttt{CAhl} or \texttt{AChl} never occur.

In fact, the end result depends only on the relative ordering of the operations that insert ``A'' and ``C'', respectively.
All other operations can be reordered without affecting the outcome.
Thus, even if the inserted strings are longer than two characters, their relative ordering only depends on the IDs of their first character.
The remaining characters follow their initial character without interleaving.

Note that there are only six possible orderings of the four operations, and not $4! = 24$, because the Lamport timestamp ordering on identifiers (as given in \S~\ref{sec:system-model}) is a linear extension of the causal order.
In this example we assume that text is typed from left to right (that is, ``A'' is always inserted before ``l'', and ``C'' is inserted before ``h'').
This implies that the ID of the operation inserting ``l'' must be greater than that of the insertion of ``A'', and likewise the ``h'' insertion must be greater than the ``C'' insertion.

\begin{figure}
\setlength{\tabcolsep}{1pt}
\begin{tabular}{ll|ll|ll}
$\mathit{id}_1, \mathsf{InsertAfter}(\mathit{id}_0), \text{``A''}$ & $\rightarrow$ \texttt{A} &
$\mathit{id}_1, \mathsf{InsertAfter}(\mathit{id}_0), \text{``A''}$ & $\rightarrow$ \texttt{A} &
$\mathit{id}_1, \mathsf{InsertAfter}(\mathit{id}_0), \text{``A''}$ & $\rightarrow$ \texttt{A} \\
$\mathit{id}_2, \mathsf{InsertAfter}(\mathit{id}_1), \text{``l''}$ & $\rightarrow$ \texttt{Al} &
$\mathit{id}_2, \mathsf{InsertAfter}(\mathit{id}_0), \text{``C''}$ & $\rightarrow$ \texttt{CA} &
$\mathit{id}_2, \mathsf{InsertAfter}(\mathit{id}_0), \text{``C''}$ & $\rightarrow$ \texttt{CA} \\
$\mathit{id}_3, \mathsf{InsertAfter}(\mathit{id}_0), \text{``C''}$ & $\rightarrow$ \texttt{CAl} &
$\mathit{id}_3, \mathsf{InsertAfter}(\mathit{id}_1), \text{``l''}$ & $\rightarrow$ \texttt{CAl} &
$\mathit{id}_3, \mathsf{InsertAfter}(\mathit{id}_2), \text{``h''}$ & $\rightarrow$ \texttt{ChA} \\
$\mathit{id}_4, \mathsf{InsertAfter}(\mathit{id}_3), \text{``h''}$ & $\rightarrow$ \texttt{ChAl} &
$\mathit{id}_4, \mathsf{InsertAfter}(\mathit{id}_2), \text{``h''}$ & $\rightarrow$ \texttt{ChAl} &
$\mathit{id}_4, \mathsf{InsertAfter}(\mathit{id}_1), \text{``l''}$ & $\rightarrow$ \texttt{ChAl} \\[6pt] \hline &&&&&\\[-6pt]
$\mathit{id}_1, \mathsf{InsertAfter}(\mathit{id}_0), \text{``C''}$ & $\rightarrow$ \texttt{C} &
$\mathit{id}_1, \mathsf{InsertAfter}(\mathit{id}_0), \text{``C''}$ & $\rightarrow$ \texttt{C} &
$\mathit{id}_1, \mathsf{InsertAfter}(\mathit{id}_0), \text{``C''}$ & $\rightarrow$ \texttt{C} \\
$\mathit{id}_2, \mathsf{InsertAfter}(\mathit{id}_0), \text{``A''}$ & $\rightarrow$ \texttt{AC} &
$\mathit{id}_2, \mathsf{InsertAfter}(\mathit{id}_0), \text{``A''}$ & $\rightarrow$ \texttt{AC} &
$\mathit{id}_2, \mathsf{InsertAfter}(\mathit{id}_1), \text{``h''}$ & $\rightarrow$ \texttt{Ch} \\
$\mathit{id}_3, \mathsf{InsertAfter}(\mathit{id}_2), \text{``l''}$ & $\rightarrow$ \texttt{AlC} &
$\mathit{id}_3, \mathsf{InsertAfter}(\mathit{id}_1), \text{``h''}$ & $\rightarrow$ \texttt{ACh} &
$\mathit{id}_3, \mathsf{InsertAfter}(\mathit{id}_0), \text{``A''}$ & $\rightarrow$ \texttt{ACh} \\
$\mathit{id}_4, \mathsf{InsertAfter}(\mathit{id}_1), \text{``h''}$ & $\rightarrow$ \texttt{AlCh} &
$\mathit{id}_4, \mathsf{InsertAfter}(\mathit{id}_2), \text{``l''}$ & $\rightarrow$ \texttt{AlCh} &
$\mathit{id}_4, \mathsf{InsertAfter}(\mathit{id}_3), \text{``l''}$ & $\rightarrow$ \texttt{AlCh} \\
\end{tabular}
\caption{All possible operation orderings when the strings ``Al'' (for ``Alice'') and ``Ch'' (for ``Charlie'') are concurrently inserted at the same position.
The operation IDs are arbitrary; we only require that $id_0 < id_1 < id_2 < id_3 < id_4$.}\label{fig:op-permutations}
\end{figure}

\begin{theorem}
    The OpSet list specification from \S~\ref{sec:datatypes} is strictly stronger than the $\mathcal{A}_\textsf{strong}$ specification of Attiya et al \cite{Attiya:2016kh}.
    That is, any algorithm that satisfies the list specification given in \S~\ref{sec:datatypes} also satisfies $\mathcal{A}_\textsf{strong}$, but the converse is not true.
\end{theorem}
\begin{proof}
    We formalise the $\mathcal{A}_\textsf{strong}$ specification with Isabelle/HOL, and produce a mechanically verified proof that every possible execution of the list specification from \S~\ref{sec:datatypes} satisfies all conditions of $\mathcal{A}_\textsf{strong}$.
    \ifarxiv
        The formal proof development is summarised in Appendix~\ref{appendix:attiya-spec}.
    \fi
    The fact that our specification is \emph{strictly} stronger follows from Theorems~\ref{thm:attiya-allows-interleaving} and~\ref{thm:no-interleaving}.
\end{proof}

\begin{theorem}
    The RGA algorithm \cite{Roh:2011dw} satisfies the OpSet list specification introduced in this paper, while Logoot \cite{Weiss:2009ht,Weiss:2010hx} and LSEQ \cite{Nedelec:2016eo,Nedelec:2013ky} do not.
\end{theorem}
\begin{proof}
    We use Isabelle/HOL to prove that RGA satisfies our specification, as described in
    \ifarxiv
        Appendix~\ref{appendix:rga}.
    \fi
    Our Isabelle/HOL implementation of RGA is based on the formalisation that we developed in previous work \cite{Gomes:2017vo,Gomes:2017gy}.
    The fact that Logoot and LSEQ do not satisfy our specification follows directly from Theorems~\ref{thm:attiya-allows-interleaving} and~\ref{thm:no-interleaving}.
\end{proof}

\section{A Replicated Tree Datatype}\label{sec:tree}

In \S~\ref{sec:datatypes} we gave an OpSet specification of a replicated object graph datatype.
In this model, every map or list object has a unique ID (namely, the ID of the $\mathsf{MakeMap}$ or $\mathsf{MakeList}$ operation that created it), and objects can reference each other using these IDs.

We now build upon this model, showing how to restrict the object graph so that it is always a tree.
A tree is a graph in which every vertex has exactly one parent (except for the root, which has no parent), and in which the parent relation has no cycles.
Tree data structures are useful in many applications: for example, file systems (consisting of directories and files) and XML or JSON documents are trees.
Branch nodes in this tree may be either maps or lists, and leaf nodes are primitive values (wrapped in a $\mathsf{MakeVal}$ operation).

\subsection{The Difficulty of a Move Operation}\label{sec:tree-difficult}

In applications that use tree-structured data, a frequently required operation is to \emph{move} a subtree to a new location within the tree.
For example:
\begin{itemize}
    \item In a filesystem, renaming a directory can be expressed as moving the directory node from the old name to the new name.
        Similarly, a directory may be moved to a new path.
    \item In vector graphics applications, several graphical objects may be grouped together as a logical unit.
        This operation can be expressed by creating a new branch node to represent the group, and then moving the individual objects to be children of that group node.
    \item In a to-do list application, users may use the order of items in the list to denote a priority order, and they may drag and drop items to change their relative order.
        Reordering items is equivalent to moving items to new locations within the list.
\end{itemize}

A move operation can be naively emulated by deleting the subtree from its old location and recreating it at the new location.
However, if two users perform this process concurrently, the resulting tree will contain two copies of the moved subtree, which would be undesirable in all of the application examples given above.
Thus, we require an \emph{atomic move} operation that does not create duplicate objects in case of concurrent moves.

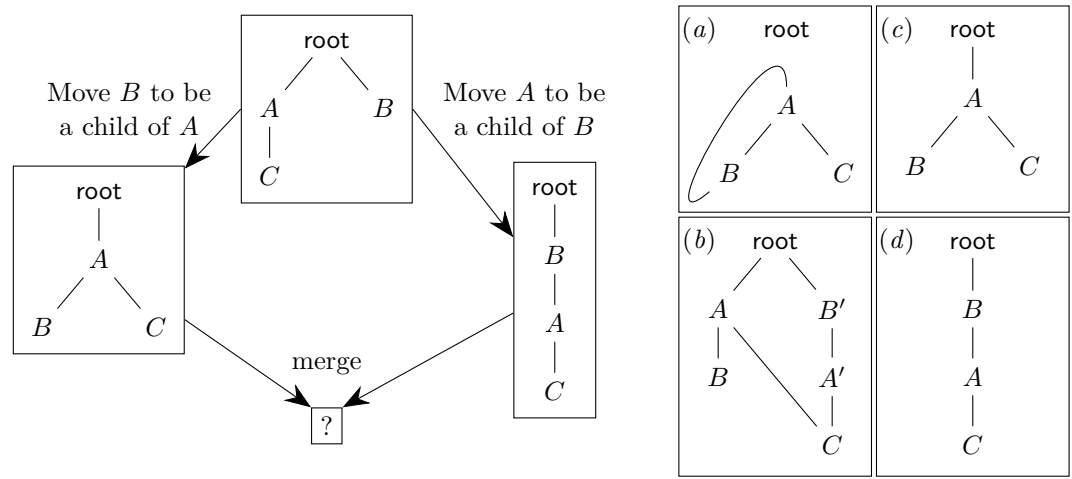
\begin{figure}
\centering
\begin{tikzpicture}
  \tikzstyle{arrow}=[draw,-{Stealth[length=3.5mm]}]
  \node [rectangle,draw] (start) at (4,4) {
      \begin{tikzpicture}
      \node {$\mathsf{root}$} [level distance=9mm] child {node {$A$} child {node {$C$}}} child {node {$B$}};
      \end{tikzpicture}
  };
  \node [rectangle,draw] (left) at (1,2) {
      \begin{tikzpicture}
      \node {$\mathsf{root}$} [level distance=9mm] child {node {$A$} child {node {$B$}} child {node {$C$}}};
      \end{tikzpicture}
  };
  \node [rectangle,draw] (right) at (7,1.6) {
      \begin{tikzpicture}
      \node {$\mathsf{root}$} [level distance=9mm] child {node {$B$} child {node {$A$} child {node {$C$}}}};
      \end{tikzpicture}
  };
  \node [rectangle,draw] (merge) at (4,-0.2) { ? };
  \node at (1.4,4.0) [text width=2.5cm,text centered] {Move $B$ to be a child of $A$};
  \node at (6.6,4.0) [text width=2.5cm,text centered] {Move $A$ to be a child of $B$};
  \draw [arrow] (start.west) -- (left);
  \draw [arrow] (start.east) -- (right);
  \draw [arrow] (left)  -- (merge.north west);
  \draw [arrow] (right) -- (merge.north east);
  \node at (4,0.6) {merge};
  \node [rectangle,draw] at (9.9,4) {
      \begin{tikzpicture}
      \useasboundingbox (-1.3,-1.3) rectangle (1,1.2);
      \node at (-1.15,1.0) {(\emph{a})};
      \node at (0,1) {$\mathsf{root}$};
      \node (a1) {$A$} [level distance=9mm] child {node (b1) {$B$}} child {node {$C$}};
      \draw (b1.south west) .. controls (-2,-2) and (0,1.5) .. (a1.north);
      \end{tikzpicture}
  };
  \node [rectangle,draw] at (9.9,0.85) {
      \begin{tikzpicture}
      \useasboundingbox (-1.15,-3.0) rectangle (1.15,0.2);
      \node at (-1.0,0.0) {(\emph{b})};
      \node {$\mathsf{root}$} [level distance=9mm] child {node (a2) {$A$} child {node {$B$}}} child {node {$B'$} child {node {$A'$} child {node (c2) {$C$}}}};
      \draw (a2) -- (c2);
      \end{tikzpicture}
  };
  \node [rectangle,draw] (start) at (12.5,4) {
      \begin{tikzpicture}
      \useasboundingbox (-1.15,-2.3) rectangle (1.15,0.2);
      \node at (-1.0,0.0) {(\emph{c})};
      \node {$\mathsf{root}$} [level distance=9mm] child {node {$A$} child {node {$B$}} child {node {$C$}}};
      \end{tikzpicture}
  };
  \node [rectangle,draw] (right) at (12.5,0.85) {
      \begin{tikzpicture}
      \useasboundingbox (-1.15,-3.0) rectangle (1.15,0.2);
      \node at (-1.0,0.0) {(\emph{d})};
      \node {$\mathsf{root}$} [level distance=9mm] child {node {$B$} child {node {$A$} child {node {$C$}}}};
      \end{tikzpicture}
  };
\end{tikzpicture}
\caption{Initially, $A$ and $B$ are siblings. $B$ is moved to be a child of $A$, while concurrently
$A$ is moved to be a child of $B$. Boxes (\emph{a}) to (\emph{d}) show possible outcomes of the merge.}\label{fig:concurrent-move}
\end{figure}

A more subtle kind of conflict is illustrated in Figure~\ref{fig:concurrent-move}.
Here, $B$ is moved to be a child of $A$, while concurrently $A$ is moved to be a child of $B$. 
If the CRDT does not take care to detect this situation, it may introduce a cycle in the merged result, as shown in Figure~\ref{fig:concurrent-move}(\emph{a});
this result is no longer a tree.
Handling such conflicting move operations is a challenging problem, and to our knowledge no existing implementation of a tree CRDT has found an adequate solution to this problem.

Several CRDT tree datatypes for XML \cite{Martin:2010ih,Nicolaescu:2015id} and JSON data \cite{Kleppmann:2016ve,Automerge,crjdt} have been developed, but to our knowledge, none of them define a move operation.
Tao et al.~\cite{Tao:2015gd} implemented a CRDT-based replicated filesystem, resolving concurrent moves with an approach illustrated in Figure~\ref{fig:concurrent-move}(\emph{b}): conflicting branch nodes (directories) are duplicated, and leaf nodes (files) may be referenced from multiple branch nodes.
Thus, Tao et al.'s data structure is strictly a DAG, not a tree.

Najafzadeh~\cite{Najafzadeh:2017vk,Najafzadeh:2018bw} also implemented a CRDT-based replicated filesystem, but chose a different approach: move operations must acquire a global lock before they can proceed, which ensures that conflicting concurrent move operations cannot occur in the first place.
This conservative approach rules out move conflicts, but the resulting datatype is not strictly a CRDT, since some operations require strongly consistent synchronisation.

\subsection{Specifying a Tree with Atomic Moves}\label{sec:tree-spec}

We now demonstrate the power of the OpSets approach by using it to define a tree CRDT with an anomaly-free atomic move operation.
Our specification rules out violations of the tree structure such as those in Figure~\ref{fig:concurrent-move}(\emph{a,b}), and concurrent moves do not duplicate tree nodes.
Moreover, our CRDT does not require any locks or global synchronisation.

When the OpSet contains conflicting move operations, our specification chooses one of them as the one that takes effect, and simply ignores the other conflicting operations.
Thus, in the example of Figure~\ref{fig:concurrent-move}, the merged outcome of the two conflicting move operations is either (\emph{c}) or (\emph{d}).
If two users concurrently move the same item to different locations, the move operation with the greater ID determines the item's final location.
However, in non-conflict situations, all concurrent move operations take effect.

We define a tree to be a restricted form of the object graph specified in \S~\ref{sec:datatypes}.
First, we require that there is a designated root object: assume that we have an operation ID $\mathsf{root}$ that is less than all other operation IDs (according to the total order on identifiers, introduced in \S~\ref{sec:system-model}).
Further assume that for any OpSet $O$ specifying a tree, we have either $(\mathsf{root},\, \mathsf{MakeList}) \in O$ or $(\mathsf{root},\, \mathsf{MakeMap}) \in O$, depending on whether the root node is a list or a map.
We define an object $x$ to be the \emph{parent} of an object $y$ if one of the values in $x$ is a reference to $y$.
The \emph{ancestor} relation is the transitive closure of the parent relation, defined using the element relation $E$:
\begin{align*}
    \mathsf{parent}(E,\, i) &=
    \begin{cases}
        \big\{ (\mathit{obj}, \mathit{val}) \mid \exists\,\mathit{id}, \mathit{key}.\;
            (\mathit{id}, \mathit{obj}, \mathit{key}, \mathit{val}) \in E \big\} & \text{if } i=1 \\
        \big\{ (x, z) \mid (x, y) \in \mathsf{parent}(E,\, i-1) \;\wedge\;
            (y, z) \in \mathsf{parent}(E,\, 1) \big\} & \text{if } i > 1
    \end{cases} \\[8pt]
    \mathsf{ancestor}(E) &= \bigcup_{i \;\geq\; 1} \mathsf{parent}(E,\, i)
\end{align*}

An object graph is a tree if the root has no parent, every non-root node has exactly one parent, and if the ancestor relation has no cycles.
We can redefine the operation interpretations from \S~\ref{sec:datatypes-interp} to preserve this tree invariant.
In fact, it is sufficient to redefine only the interpretation of $\mathsf{Assign}$, and to leave the interpretation of the other five operation types unchanged:
\begin{align*}
    \mathsf{interp}&\big[(E,\, L),\; (\mathit{id},\, \mathsf{Assign}(\mathit{obj}, \mathit{key}, \mathit{val}, \mathit{prev})) \big] \;=\\
    & \left\{
    \arraycolsep=0pt \def\arraystretch{1.5}
    \begin{array}{l}
        (E,\, L) \qquad \text{if } (\mathit{val},\, \mathit{obj}) \in \mathsf{ancestor}(E) \\[2pt]
        \Big( \big\{ (\mathit{id}', \mathit{obj}', \mathit{key}', \mathit{val}') \in E \mid
        \mathit{id}' \notin \mathit{prev} \wedge \mathit{val}' \neq \mathit{val} \big\} \;\cup\;
        \big\{ (\mathit{id}, \mathit{obj}, \mathit{key}, \mathit{val}) \big\},\; L \Big) \\
        \hphantom{(E,\, L)} \qquad \text{if } (\mathit{val},\, \mathit{obj}) \notin \mathsf{ancestor}(E)
    \end{array} \right.
\end{align*}

This definition differs in two ways from that in \S~\ref{sec:datatypes-interp}.
Firstly, the operation has no effect if $\mathit{val}$ is already an ancestor of the proposed parent $\mathit{obj}$, since the operation would otherwise introduce a cycle.
Secondly, any existing tuple in $E$ that references the same value $\mathit{val}$ is removed, preserving the invariant that every non-root node must have exactly one parent.

This interpretation of $\mathsf{Assign}$ performs an atomic move whenever $\mathit{val}$ is the ID of an existing object in the tree; in that case, it is moved from its existing position to the key $\mathit{key}$ in the object $\mathit{obj}$.
If $\mathit{val}$ does not currently exist in the tree (e.g.\ because it has just been created), the operation behaves like conventional assignment.

\section{Related Work}\label{sec:relwork}

\subsection{Interpretation of Operation Sequences}\label{sec:op-sequences}

The general idea of establishing a total order of operations, and executing them in that order, appears in many areas of computing:
for example, in the state machine approach to replication \cite{Schneider:1990vy},
the event sourcing approach to data modelling \cite{Vernon:2013ww},
write-ahead logs for crash recovery \cite{Mohan:1992fe},
serializable transactions \cite{Davidson:1985hv},
and scalable multicore data structures \cite{BoydWickizer:2014uz}.
However, beneath the superficial similarity of these approaches there are important differences that need to be distinguished.

As discussed in \S~\ref{sec:order-change}, many of these systems rely on the property that after some operation is executed, all subsequent operations will appear \emph{after} it in the total order.
In other words, the operation sequence is an append-only log, and new operations never need to be inserted ahead of an existing operation in the total order.
This is a very strong property: in the context of a distributed system, it requires an atomic broadcast (or total order broadcast) protocol \cite{Defago:2004ji}, which is equivalent to solving distributed consensus \cite{Chandra:1996cp}.
This class of protocols requires communication with a quorum of nodes in order to make progress \cite{Howard:2016tz}, and it cannot guarantee progress in a fully asynchronous setting \cite{Fischer:1985tt}.

By contrast, the sequential OpSet interpretation of \S~\ref{sec:op-serial} does not require atomic broadcast because it allows operations to be added to the OpSet in any order, and it assigns operation IDs without any coordination.
Few systems use this approach; the most closely related prior work are the Bayou system \cite{Terry:1995dn}, which executes tentative transactions deterministically in timestamp order, and Burckhardt's \emph{standard conflict resolution} \cite[\S~4.3.3]{Burckhardt:2014hy}.
Both of these share the OpSet approach's characteristic that operations with a higher ID need to be undone and re-applied when a new operation with a lower ID is received.

Our contribution in this paper is to formulate the OpSet approach more generally as a tool for specifying and reasoning about complex replicated data structures, such as lists and trees.
Our work is the first to use this approach in mechanised proofs, in which we show that a non-OpSet list CRDT (RGA) satisfies an OpSet-based specification, and prove the absence of the interleaving anomaly in Figure~\ref{fig:bad-merge}.

Baquero et al.~\cite{Baquero:2014ed} and Grishchenko~\cite{Grishchenko:2014eh} have proposed representing CRDTs in terms of a partially-ordered log of operations, where the partial order captures the causal relationships between operations.
The OpSet approach can be seen as a variant of this idea, in which we define the total order on identifiers to be a linear extension of the partial order.

\subsection{Specification and Verification of Replicated Datatypes}

Algorithms for collaboratively editing a shared data structure have been the topic of active research for approximately 30 years, under the headings of Operational Transformation \cite{Ellis:1989ue,Ressel:1996wx,Sun:1998vf,Oster:2006tr} and CRDTs \cite{Shapiro:2011wy,Shapiro:2011un}.
However, throughout this time, the exact consistency properties provided by the algorithms have been somewhat unclear.
For example, Sun et al.~\cite{Sun:1998un} identified three desirable properties that they articulated informally: \emph{convergence}, \emph{causality preservation}, and \emph{intention preservation}.
While the definition of the first two properties is fairly unambiguous, the definition of ``intention preservation'' leaves much more room for interpretation.
Efforts to formally specify and verify the semantics of replicated datatypes have replaced such informal statements with precise consistency properties.

Burckhardt et al.~\cite{Burckhardt:2014ft} provide a wide-ranging formal account of CRDTs, covering their specification, verification, and optimality, with the semantics of an operation on a replicated datatype given as a function of the operation, $o$, and a \emph{operation context}---the set of operations visible to a node at the time that $o$ was received.
Our OpSets can be seen as an explicitly executable variation on this idea: nodes record all operations that they have ever received in a monotonically growing set, and the interpretation function builds the result ``bottom up'' in a fold-like operation.
In contrast to Burckhardt et al., who focus on applying their techniques to set and counter datatypes, we apply our approach to the specification of lists, maps, and trees, using our OpSets as a tool for designing new replicated datatypes---including those previously thought impossible, such as our replicated tree with atomic move.
Gotsman et al.~\cite{DBLP:conf/popl/GotsmanYFNS16} extend Burckhardt et al.'s formalism to reason about hybrid consistency models, providing a modular proof rule inspired by permissions-based logics to enforce an integrity invariant for a given consistency model.

Bieniusa et al.~\cite{Bieniusa:2012gt} articulate a \emph{principle of permutation equivalence} that partially specifies the expected semantics of replicated datatypes, but which leaves some combinations of operations unspecified.
Zeller et al.~\cite{Zeller:2014fl} formalise counters, registers, and sets using Isabelle/HOL and provide mechanised proofs of their correctness.
Attiya et al.~\cite{Attiya:2016kh} give two specifications of collaborative text editing ($\mathcal{A}_\textsf{strong}$ and $\mathcal{A}_\textsf{weak}$), prove that the RGA CRDT \cite{Roh:2011dw} satisfies $\mathcal{A}_\textsf{strong}$, and conjecture that the Operational Transformation algorithm Jupiter \cite{Nichols:1995fd} satisfies $\mathcal{A}_\textsf{weak}$.
Wei et al.~\cite{Wei:2017tg} complete the proof that Jupiter satisfies $\mathcal{A}_\textsf{weak}$.

In our prior work \cite{Gomes:2017gy} we establish a formal verification framework for CRDTs in Isabelle/HOL, and verify the strong eventual consistency properties (in particular, convergence) of a list, set, and counter datatype.
The Isabelle implementation of RGA we use in \S~\ref{sec:datatypes} is based on this work \cite{Gomes:2017vo}.
However, this work does not specify the datatype semantics beyond the convergence property.

Gaducci et al.~\cite{DBLP:conf/coordination/GadducciMR17} develop a semantics for replicated datatypes, placing a focus on compositionality, where a replicated datatype is modelled as a function from labelled directed acyclic graphs of events to sets of values, with each value in this set potentially observable at a node under different ordering of events observed at that node.
A notion of behavioural \emph{refinement} for replicated datatypes induced by set inclusion is also defined, along with a generalisation of their relational semantics to a categorical one.

Mukund et al.~\cite{DBLP:conf/atva/MukundRS15} use traces to provide bounded declarative specifications of CRDTs and show how Counter Example Guided Abstract Refinement (CEGAR) can be used to automatically verify a reference CRDT implementation against its bounded specification.

\subsection{Collaborative Tree Datatypes}

For collaborative editing of tree data structures, several CRDTs \cite{Martin:2010ih,Kleppmann:2016ve} and Operational Transformation algorithms \cite{Jungnickel:2016cb,Ignat:2003jy,Davis:2002iv} have been proposed.
However, most of them only consider insertion and deletion of tree nodes, but do not support a move operation.

As explained in \S~\ref{sec:tree}, supporting an operation that can move a subtree to a new location within a tree introduces new conflicts that need to be handled.
Ahmed-Nacer et al.~\cite{AhmedNacer:2012us} survey approaches to handling these conflicts without providing concrete algorithms.
Tao et al.~\cite{Tao:2015gd} propose handling conflicting move operations by allowing the same object to appear in more than one location; thus, their datatype is strictly a DAG, not a tree.

Najafzadeh~\cite{Najafzadeh:2017vk,Najafzadeh:2018bw} asserts that concurrent move operations on a tree cannot safely be implemented in a CRDT, since the precondition of a move operation is not stable.
Najafzadeh suggests the use of locks to globally synchronise move operations, preventing a scenario such as that in Figure~\ref{fig:concurrent-move} from ever occurring.
However, the resulting datatype is not strictly a CRDT, since some operations require strongly consistent synchronisation.

To our knowledge, our move semantics specified in \S~\ref{sec:tree} is the first definition of such an operation on a fully asynchronous tree CRDT.
We avoid the apparent contradiction with Najafzadeh's assertion by evaluating the precondition $(\mathit{val},\, \mathit{obj}) \notin \mathsf{ancestor}(E)$ at the same time as applying the operation, rather than at the time when the operation is generated, and by applying all operations in the OpSet in a deterministic order.

\section{Conclusion}

In this work we have introduced Operation Sets (OpSets), a simple but powerful approach for specifying the semantics of replicated datatypes.
We specified a variety of common, composable replicated datatypes in the OpSets model, and used Isabelle/HOL to formally reason about their properties.
We have used this specification to highlight an interleaving anomaly that affects some existing collaborative text editing algorithms, and proved that the RGA algorithm satisfies our list specification.
Finally, we demonstrated how the OpSet model to can be used to develop new replication algorithms, and we introduced a specification for an atomic move operation in a tree CRDT.

The OpSets approach is an executable specification that precisely defines the permitted states of a replica after some set of updates have been applied.
In this paper we have used a sequential OpSet interpretation: operations are applied in strict ascending order of ID.
This property is very useful as it trivially ensures convergence, and it simplifies reasoning about specifications and invariants for CRDTs.
In contrast, the traditional approach to defining CRDTs requires operations to be commutative, increasing their complexity.
In proving that RGA satisfies our list specification, we demonstrate a correspondence between sequential specification and commutative implementation; for future work it will be interesting to further explore this correspondence for other datatypes.
In particular, we hypothesise that it is possible to derive a tree CRDT with a commutative move operation from the specification in \S~\ref{sec:tree}, which could then be used to implement a distributed peer-to-peer file system.

Although we focussed on sequential OpSet interpretations in this paper, note that any deterministic function can be used as interpretation function $\llbracket-\rrbracket$.
In particular, one can view the OpSet as a \emph{database of facts}, containing all changes ever made to the shared data, and the interpretation function as a \emph{query} over this database.
The resulting datatype is then a \emph{materialized view} in database terminology.
When new operations are added to an OpSet $O$, computing the corresponding change to $\llbracket O \rrbracket$ is a materialized view maintenance problem, for which optimised algorithms have been developed \cite{Gupta:1999uz}.
We hypothesise that these techniques can be applied to replicated datatypes, allowing efficient CRDT implementations to be derived from an OpSet-based specification.

\subsection*{Acknowledgements}

The authors wish to acknowledge the support of The Boeing Company,
the EPSRC ``REMS: Rigorous Engineering for Mainstream Systems'' programme grant (EP/K008528), and
the EPSRC ``Interdisciplinary Centre for Finding, Understanding and Countering Crime in the Cloud'' grant (EP/M020320).
We thank Nathan Chong, Peter Sewell, and KC Sivaramakrishnan for their helpful feedback on this paper.

\bibliographystyle{plainurl}
\bibliography{references}{}

\begin{thebibliography}{10}

\bibitem{AhmedNacer:2011ke}
Mehdi Ahmed-Nacer, Claudia-Lavinia Ignat, G{\'e}rald Oster, Hyun-Gul Roh, and
  Pascal Urso.
\newblock Evaluating {CRDTs} for real-time document editing.
\newblock In {\em 11th ACM Symposium on Document Engineering (DocEng)}, pages
  103--112, September 2011.
\newblock \href {http://dx.doi.org/10.1145/2034691.2034717}
  {\path{doi:10.1145/2034691.2034717}}.

\bibitem{AhmedNacer:2012us}
Mehdi Ahmed-Nacer, St{\'e}phane Martin, and Pascal Urso.
\newblock File system on {CRDT}.
\newblock Technical Report RR-8027, INRIA, July 2012.
\newblock URL: \url{https://hal.inria.fr/hal-00720681/}.

\bibitem{ReplicationBenchmark}
Mehdi Ahmed-Nacer, G{\'e}rald Oster, and Pascal Urso.
\newblock Java benchmarker of optimistic replication algorithms.
\newblock URL: \url{https://github.com/PascalUrso/ReplicationBenchmark}.

\bibitem{Attiya:2016kh}
Hagit Attiya, Sebastian Burckhardt, Alexey Gotsman, Adam Morrison, Hongseok
  Yang, and Marek Zawirski.
\newblock Specification and complexity of collaborative text editing.
\newblock In {\em ACM Symposium on Principles of Distributed Computing (PODC)},
  pages 259--268, July 2016.
\newblock \href {http://dx.doi.org/10.1145/2933057.2933090}
  {\path{doi:10.1145/2933057.2933090}}.

\bibitem{Bailis:2013jc}
Peter Bailis and Ali Ghodsi.
\newblock Eventual consistency today: Limitations, extensions, and beyond.
\newblock {\em ACM Queue}, 11(3), March 2013.
\newblock \href {http://dx.doi.org/10.1145/2460276.2462076}
  {\path{doi:10.1145/2460276.2462076}}.

\bibitem{Baquero:2014ed}
Carlos Baquero, Paulo~S{\'e}rgio Almeida, and Ali Shoker.
\newblock Making operation-based {CRDTs} operation-based.
\newblock In {\em 14th IFIP International Conference on Distributed
  Applications and Interoperable Systems (DAIS)}, pages 126--140, June 2014.
\newblock \href {http://dx.doi.org/10.1007/978-3-662-43352-2_11}
  {\path{doi:10.1007/978-3-662-43352-2_11}}.

\bibitem{Bieniusa:2012gt}
Annette Bieniusa, Marek Zawirski, Nuno Pregui{\c c}a, Marc Shapiro, Carlos
  Baquero, Valter Balegas, and S{\'e}rgio Duarte.
\newblock Brief announcement: Semantics of eventually consistent replicated
  sets.
\newblock In {\em 26th International Symposium on Distributed Computing
  (DISC)}, pages 441--442, October 2012.
\newblock \href {http://dx.doi.org/10.1007/978-3-642-33651-5_48}
  {\path{doi:10.1007/978-3-642-33651-5_48}}.

\bibitem{BoydWickizer:2014uz}
Silas Boyd-Wickizer, M~Frans Kaashoek, Robert Morris, and Nickolai Zeldovich.
\newblock {OpLog}: a library for scaling update-heavy data structures.
\newblock Technical Report MIT-CSAIL-TR-2014-019, MIT CSAIL, September 2014.
\newblock URL: \url{http://hdl.handle.net/1721.1/89653}.

\bibitem{Burckhardt:2014hy}
Sebastian Burckhardt.
\newblock Principles of eventual consistency.
\newblock {\em Foundations and Trends in Programming Languages}, 1(1-2):1--150,
  October 2014.
\newblock \href {http://dx.doi.org/10.1561/2500000011}
  {\path{doi:10.1561/2500000011}}.

\bibitem{Burckhardt:2014ft}
Sebastian Burckhardt, Alexey Gotsman, Hongseok Yang, and Marek Zawirski.
\newblock Replicated data types: Specification, verification, optimality.
\newblock In {\em 41st ACM SIGPLAN-SIGACT Symposium on Principles of
  Programming Languages (POPL)}, pages 271--284, January 2014.
\newblock \href {http://dx.doi.org/10.1145/2535838.2535848}
  {\path{doi:10.1145/2535838.2535848}}.

\bibitem{Chandra:1996cp}
Tushar~Deepak Chandra and Sam Toueg.
\newblock Unreliable failure detectors for reliable distributed systems.
\newblock {\em Journal of the ACM}, 43(2):225--267, March 1996.
\newblock \href {http://dx.doi.org/10.1145/226643.226647}
  {\path{doi:10.1145/226643.226647}}.

\bibitem{LSEQTree}
{Chat-Wane}.
\newblock {LSEQTree}.
\newblock URL: \url{https://github.com/Chat-Wane/LSEQTree}.

\bibitem{Davidson:1985hv}
Susan~B Davidson, Hector Garcia-Molina, and Dale Skeen.
\newblock Consistency in partitioned networks.
\newblock {\em ACM Computing Surveys}, 17(3):341--370, September 1985.
\newblock \href {http://dx.doi.org/10.1145/5505.5508}
  {\path{doi:10.1145/5505.5508}}.

\bibitem{Davis:2002iv}
Aguido~Horatio Davis, Chengzheng Sun, and Junwei Lu.
\newblock Generalizing operational transformation to the {Standard General
  Markup Language}.
\newblock In {\em ACM Conference on Computer Supported Cooperative Work
  (CSCW)}, pages 58--67, November 2002.
\newblock \href {http://dx.doi.org/10.1145/587078.587088}
  {\path{doi:10.1145/587078.587088}}.

\bibitem{Defago:2004ji}
Xavier D{\'e}fago, Andr{\'e} Schiper, and P{\'e}ter Urb{\'a}n.
\newblock Total order broadcast and multicast algorithms: Taxonomy and survey.
\newblock {\em ACM Computing Surveys}, 36(4):372--421, December 2004.
\newblock \href {http://dx.doi.org/10.1145/1041680.1041682}
  {\path{doi:10.1145/1041680.1041682}}.

\bibitem{Ellis:1989ue}
Clarence Ellis and S~J Gibbs.
\newblock Concurrency control in groupware systems.
\newblock In {\em ACM International Conference on Management of Data (SIGMOD)},
  pages 399--407, May 1989.
\newblock \href {http://dx.doi.org/10.1145/67544.66963}
  {\path{doi:10.1145/67544.66963}}.

\bibitem{Fischer:1985tt}
Michael~J Fischer, Nancy~A Lynch, and Michael~S Paterson.
\newblock Impossibility of distributed consensus with one faulty process.
\newblock {\em Journal of the ACM}, 32(2):374--382, April 1985.
\newblock \href {http://dx.doi.org/10.1145/3149.214121}
  {\path{doi:10.1145/3149.214121}}.

\bibitem{DBLP:conf/coordination/GadducciMR17}
Fabio Gadducci, Hern{\'{a}}n~C. Melgratti, and Christian Rold{\'{a}}n.
\newblock A denotational view of replicated data types.
\newblock In {\em 19th International Conference on Coordination Models and
  Languages (COORDINATION)}, pages 138--156, June 2017.
\newblock \href {http://dx.doi.org/10.1007/978-3-319-59746-1_8}
  {\path{doi:10.1007/978-3-319-59746-1_8}}.

\bibitem{Gilbert:2002il}
Seth Gilbert and Nancy~A Lynch.
\newblock Brewer's conjecture and the feasibility of consistent, available,
  partition-tolerant web services.
\newblock {\em ACM SIGACT News}, 33(2):51--59, 2002.
\newblock \href {http://dx.doi.org/10.1145/564585.564601}
  {\path{doi:10.1145/564585.564601}}.

\bibitem{Gomes:2017vo}
Victor B~F Gomes, Martin Kleppmann, Dominic~P Mulligan, and Alastair~R
  Beresford.
\newblock A framework for establishing strong eventual consistency for
  conflict-free replicated data types.
\newblock {\em Archive of Formal Proofs}, July 2017.
\newblock URL: \url{http://isa-afp.org/entries/CRDT.html}.

\bibitem{Gomes:2017gy}
Victor B~F Gomes, Martin Kleppmann, Dominic~P Mulligan, and Alastair~R
  Beresford.
\newblock Verifying strong eventual consistency in distributed systems.
\newblock {\em Proceedings of the ACM on Programming Languages (PACMPL)},
  1(OOPSLA), October 2017.
\newblock \href {http://dx.doi.org/10.1145/3133933}
  {\path{doi:10.1145/3133933}}.

\bibitem{DBLP:conf/popl/GotsmanYFNS16}
Alexey Gotsman, Hongseok Yang, Carla Ferreira, Mahsa Najafzadeh, and Marc
  Shapiro.
\newblock {`Cause I'm} strong enough: reasoning about consistency choices in
  distributed systems.
\newblock In {\em 43rd ACM SIGPLAN-SIGACT Symposium on Principles of
  Programming Languages (POPL)}, pages 371--384, January 2016.
\newblock \href {http://dx.doi.org/10.1145/2837614.2837625}
  {\path{doi:10.1145/2837614.2837625}}.

\bibitem{Grishchenko:2014eh}
Victor Grishchenko.
\newblock {Citrea} and {Swarm}: Partially ordered op logs in the browser.
\newblock In {\em 1st Workshop on Principles and Practice of Eventual
  Consistency (PaPEC)}, April 2014.
\newblock \href {http://dx.doi.org/10.1145/2596631.2596641}
  {\path{doi:10.1145/2596631.2596641}}.

\bibitem{Gupta:1999uz}
Ashish Gupta and Inderpal~Singh Mumick.
\newblock {\em Materialized Views: Techniques, Implementations, and
  Applications}.
\newblock MIT Press, May 1999.

\bibitem{DBLP:conf/types/HaftmannW08}
Florian Haftmann and Makarius Wenzel.
\newblock Local theory specifications in {Isabelle/Isar}.
\newblock In {\em International Workshop on Types for Proofs and Programs
  (TYPES)}, pages 153--168, 2008.
\newblock \href {http://dx.doi.org/10.1007/978-3-642-02444-3_10}
  {\path{doi:10.1007/978-3-642-02444-3_10}}.

\bibitem{Herlihy:1990jq}
Maurice~P Herlihy and Jeannette~M Wing.
\newblock Linearizability: A correctness condition for concurrent objects.
\newblock {\em ACM Transactions on Programming Languages and Systems (TOPLAS)},
  12(3):463--492, July 1990.
\newblock \href {http://dx.doi.org/10.1145/78969.78972}
  {\path{doi:10.1145/78969.78972}}.

\bibitem{Howard:2016tz}
Heidi Howard, Dahlia Malkhi, and Alexander Spiegelman.
\newblock {Flexible Paxos}: Quorum intersection revisited.
\newblock In {\em 20th International Conference on Principles of Distributed
  Systems (OPODIS)}, December 2016.
\newblock \href {http://dx.doi.org/10.4230/LIPIcs.OPODIS.2016.25}
  {\path{doi:10.4230/LIPIcs.OPODIS.2016.25}}.

\bibitem{Ignat:2003jy}
Claudia-Lavinia Ignat and Moira~C Norrie.
\newblock Customizable collaborative editor relying on {treeOPT} algorithm.
\newblock In {\em 8th European Conference on Computer-Supported Cooperative
  Work (ECSCW)}, pages 315--334, September 2003.
\newblock \href {http://dx.doi.org/10.1007/978-94-010-0068-0_17}
  {\path{doi:10.1007/978-94-010-0068-0_17}}.

\bibitem{Jungnickel:2016cb}
Tim Jungnickel and Tobias Herb.
\newblock Simultaneous editing of {JSON} objects via operational
  transformation.
\newblock In {\em 31st Annual ACM Symposium on Applied Computing (SAC)}, pages
  812--815, April 2016.
\newblock \href {http://dx.doi.org/10.1145/2851613.2852003}
  {\path{doi:10.1145/2851613.2852003}}.

\bibitem{DBLP:conf/tphol/KammullerWP99}
Florian Kamm{\"{u}}ller, Markus Wenzel, and Lawrence~C. Paulson.
\newblock Locales - {A} sectioning concept for {Isabelle}.
\newblock In {\em 12th International Conference on Theorem Proving in Higher
  Order Logics (TPHOLs)}, pages 149--166, 1999.
\newblock \href {http://dx.doi.org/10.1007/3-540-48256-3_11}
  {\path{doi:10.1007/3-540-48256-3_11}}.

\bibitem{Kleppmann:2017wj}
Martin Kleppmann.
\newblock {\em Designing Data-Intensive Applications}.
\newblock O'Reilly Media, April 2017.

\bibitem{Kleppmann:2016ve}
Martin Kleppmann and Alastair~R Beresford.
\newblock A conflict-free replicated {JSON} datatype.
\newblock {\em IEEE Transactions on Parallel and Distributed Systems},
  28(10):2733--2746, April 2017.
\newblock \href {http://dx.doi.org/10.1109/TPDS.2017.2697382}
  {\path{doi:10.1109/TPDS.2017.2697382}}.

\bibitem{AFP}
Martin Kleppmann, Victor B.~F. Gomes, Dominic~P. Mulligan, and Alastair~R.
  Beresford.
\newblock {OpSets}: Sequential specifications for replicated datatypes (proof
  document), May 2018.
\newblock URL: \url{https://www.isa-afp.org/entries/OpSets.html}.

\bibitem{Automerge}
Martin Kleppmann, Peter van Hardenberg, and Orion Henry.
\newblock {Automerge}.
\newblock URL: \url{https://github.com/automerge/automerge}.

\bibitem{Lamport:1978jq}
Leslie Lamport.
\newblock Time, clocks, and the ordering of events in a distributed system.
\newblock {\em Communications of the ACM}, 21(7):558--565, July 1978.
\newblock \href {http://dx.doi.org/10.1145/359545.359563}
  {\path{doi:10.1145/359545.359563}}.

\bibitem{Leach:2005hm}
Paul~J Leach, Michael Mealling, and Rich Salz.
\newblock {A Universally Unique IDentifier (UUID) URN} namespace.
\newblock IETF Standards Track, RFC 4122, July 2005.
\newblock \href {http://dx.doi.org/10.17487/rfc4122}
  {\path{doi:10.17487/rfc4122}}.

\bibitem{Martin:2010ih}
St{\'e}phane Martin, Pascal Urso, and St{\'e}phane Weiss.
\newblock Scalable {XML} collaborative editing with undo.
\newblock In {\em On the Move to Meaningful Internet Systems}, pages 507--514,
  October 2010.
\newblock \href {http://dx.doi.org/10.1007/978-3-642-16934-2_37}
  {\path{doi:10.1007/978-3-642-16934-2_37}}.

\bibitem{Mohan:1992fe}
C~Mohan, Don Haderle, Bruce Lindsay, Hamid Pirahesh, and Peter Schwarz.
\newblock {ARIES}: A transaction recovery method supporting fine-granularity
  locking and partial rollbacks using write-ahead logging.
\newblock {\em ACM Transactions on Database Systems (TODS)}, 17(1):94--162,
  March 1992.
\newblock \href {http://dx.doi.org/10.1145/128765.128770}
  {\path{doi:10.1145/128765.128770}}.

\bibitem{DBLP:conf/atva/MukundRS15}
Madhavan Mukund, Gautham~Shenoy R., and S.~P. Suresh.
\newblock Effective verification of replicated data types using later
  appearance records {(LAR)}.
\newblock In {\em 13th International Symposium on Automated Technology for
  Verification and Analysis (ATVA)}, pages 293--308, October 2015.
\newblock \href {http://dx.doi.org/10.1007/978-3-319-24953-7_23}
  {\path{doi:10.1007/978-3-319-24953-7_23}}.

\bibitem{Najafzadeh:2017vk}
Mahsa Najafzadeh.
\newblock {\em The Analysis and Co-design of Weakly-Consistent Applications}.
\newblock PhD thesis, Universit{\'e} Pierre et Marie Curie, August 2016.
\newblock URL: \url{https://tel.archives-ouvertes.fr/tel-01351187v1}.

\bibitem{Najafzadeh:2018bw}
Mahsa Najafzadeh, Marc Shapiro, and Patrick Eugster.
\newblock Co-design and verification of an available file system.
\newblock In {\em 19th International Conference on Verification, Model
  Checking, and Abstract Interpretation (VMCAI)}, pages 358--381, January 2018.
\newblock \href {http://dx.doi.org/10.1007/978-3-319-73721-8_17}
  {\path{doi:10.1007/978-3-319-73721-8_17}}.

\bibitem{Nedelec:2016eo}
Brice N{\'e}delec, Pascal Molli, and Achour Mostefaoui.
\newblock {CRATE}: Writing stories together with our browsers.
\newblock In {\em 25th International World Wide Web Conference (WWW)}, pages
  231--234, April 2016.
\newblock \href {http://dx.doi.org/10.1145/2872518.2890539}
  {\path{doi:10.1145/2872518.2890539}}.

\bibitem{Nedelec:2013ky}
Brice N{\'e}delec, Pascal Molli, Achour Mostefaoui, and Emmanuel Desmontils.
\newblock {LSEQ}: an adaptive structure for sequences in distributed
  collaborative editing.
\newblock In {\em 13th ACM Symposium on Document Engineering (DocEng)}, pages
  37--46, September 2013.
\newblock \href {http://dx.doi.org/10.1145/2494266.2494278}
  {\path{doi:10.1145/2494266.2494278}}.

\bibitem{Nichols:1995fd}
David~A Nichols, Pavel Curtis, Michael Dixon, and John Lamping.
\newblock High-latency, low-bandwidth windowing in the {Jupiter} collaboration
  system.
\newblock In {\em 8th Annual ACM Symposium on User Interface Software and
  Technology (UIST)}, pages 111--120, November 1995.
\newblock \href {http://dx.doi.org/10.1145/215585.215706}
  {\path{doi:10.1145/215585.215706}}.

\bibitem{Nicolaescu:2015id}
Petru Nicolaescu, Kevin Jahns, Michael Derntl, and Ralf Klamma.
\newblock {Yjs}: A framework for near real-time {P2P} shared editing on
  arbitrary data types.
\newblock In {\em 15th International Conference on Web Engineering (ICWE)},
  June 2015.
\newblock \href {http://dx.doi.org/10.1007/978-3-319-19890-3_55}
  {\path{doi:10.1007/978-3-319-19890-3_55}}.

\bibitem{DBLP:books/sp/NipkowK14}
Tobias Nipkow and Gerwin Klein.
\newblock {\em Concrete Semantics - With Isabelle/HOL}.
\newblock Springer, 2014.
\newblock \href {http://dx.doi.org/10.1007/978-3-319-10542-0}
  {\path{doi:10.1007/978-3-319-10542-0}}.

\bibitem{Oster:2006tr}
G{\'e}rald Oster, Pascal Molli, Pascal Urso, and Abdessamad Imine.
\newblock Tombstone transformation functions for ensuring consistency in
  collaborative editing systems.
\newblock In {\em 2nd International Conference on Collaborative Computing
  (CollaborateCom)}, 2006.
\newblock \href {http://dx.doi.org/10.1109/COLCOM.2006.361867}
  {\path{doi:10.1109/COLCOM.2006.361867}}.

\bibitem{Ressel:1996wx}
Matthias Ressel, Doris Nitsche-Ruhland, and Rul Gunzenh{\"a}uer.
\newblock An integrating, transformation-oriented approach to concurrency
  control and undo in group editors.
\newblock In {\em ACM Conference on Computer Supported Cooperative Work
  (CSCW)}, pages 288--297, November 1996.
\newblock \href {http://dx.doi.org/10.1145/240080.240305}
  {\path{doi:10.1145/240080.240305}}.

\bibitem{Roh:2011dw}
Hyun-Gul Roh, Myeongjae Jeon, Jin-Soo Kim, and Joonwon Lee.
\newblock Replicated abstract data types: Building blocks for collaborative
  applications.
\newblock {\em Journal of Parallel and Distributed Computing}, 71(3):354--368,
  2011.
\newblock \href {http://dx.doi.org/10.1016/j.jpdc.2010.12.006}
  {\path{doi:10.1016/j.jpdc.2010.12.006}}.

\bibitem{Schneider:1990vy}
Fred~B Schneider.
\newblock Implementing fault-tolerant services using the state machine
  approach: A tutorial.
\newblock {\em ACM Computing Surveys}, 22(4):299--319, December 1990.
\newblock \href {http://dx.doi.org/10.1145/98163.98167}
  {\path{doi:10.1145/98163.98167}}.

\bibitem{Shapiro:2011wy}
Marc Shapiro, Nuno Pregui{\c c}a, Carlos Baquero, and Marek Zawirski.
\newblock A comprehensive study of convergent and commutative replicated data
  types.
\newblock Technical Report 7506, INRIA, 2011.
\newblock URL: \url{http://hal.inria.fr/inria-00555588/}.

\bibitem{Shapiro:2011un}
Marc Shapiro, Nuno Pregui{\c c}a, Carlos Baquero, and Marek Zawirski.
\newblock Conflict-free replicated data types.
\newblock In {\em 13th International Symposium on Stabilization, Safety, and
  Security of Distributed Systems (SSS)}, pages 386--400, October 2011.
\newblock \href {http://dx.doi.org/10.1007/978-3-642-24550-3_29}
  {\path{doi:10.1007/978-3-642-24550-3_29}}.

\bibitem{Sun:1998vf}
Chengzheng Sun and Clarence Ellis.
\newblock Operational transformation in real-time group editors: Issues,
  algorithms, and achievements.
\newblock In {\em ACM Conference on Computer Supported Cooperative Work
  (CSCW)}, pages 59--68, November 1998.
\newblock \href {http://dx.doi.org/10.1145/289444.289469}
  {\path{doi:10.1145/289444.289469}}.

\bibitem{Sun:1998un}
Chengzheng Sun, Xiaohua Jia, Yanchun Zhang, Yun Yang, and David Chen.
\newblock Achieving convergence, causality preservation, and intention
  preservation in real-time cooperative editing systems.
\newblock {\em ACM Transactions on Computer-Human Interaction (TOCHI)},
  5(1):63--108, 1998.
\newblock \href {http://dx.doi.org/10.1145/274444.274447}
  {\path{doi:10.1145/274444.274447}}.

\bibitem{Tao:2015gd}
Vinh Tao, Marc Shapiro, and Vianney Rancurel.
\newblock Merging semantics for conflict updates in geo-distributed file
  systems.
\newblock In {\em 8th ACM International Systems and Storage Conference
  (SYSTOR)}, May 2015.
\newblock \href {http://dx.doi.org/10.1145/2757667.2757683}
  {\path{doi:10.1145/2757667.2757683}}.

\bibitem{Terry:1994fp}
Douglas~B Terry, Alan~J Demers, Karin Petersen, Mike~J Spreitzer, Marvin~M
  Theimer, and Brent~B Welch.
\newblock Session guarantees for weakly consistent replicated data.
\newblock In {\em 3rd International Conference on Parallel and Distributed
  Information Systems (PDIS)}, pages 140--149, September 1994.
\newblock \href {http://dx.doi.org/10.1109/PDIS.1994.331722}
  {\path{doi:10.1109/PDIS.1994.331722}}.

\bibitem{Terry:1995dn}
Douglas~B Terry, Marvin~M Theimer, Karin Petersen, Alan~J Demers, Mike~J
  Spreitzer, and Carl~H Hauser.
\newblock Managing update conflicts in {Bayou}, a weakly connected replicated
  storage system.
\newblock In {\em 15th ACM Symposium on Operating Systems Principles (SOSP)},
  pages 172--182, December 1995.
\newblock \href {http://dx.doi.org/10.1145/224056.224070}
  {\path{doi:10.1145/224056.224070}}.

\bibitem{crjdt}
Frank~S Thomas.
\newblock {crjdt}: A conflict-free replicated {JSON} datatype in {Scala}.
\newblock URL: \url{https://github.com/fthomas/crjdt}.

\bibitem{Vernon:2013ww}
Vaughn Vernon.
\newblock {\em {Implementing Domain-Driven Design}}.
\newblock Addison-Wesley Professional, February 2013.

\bibitem{Vogels:2009ca}
Werner Vogels.
\newblock Eventually consistent.
\newblock {\em Communications of the ACM}, 52(1):40--44, January 2009.
\newblock \href {http://dx.doi.org/10.1145/1435417.1435432}
  {\path{doi:10.1145/1435417.1435432}}.

\bibitem{Wei:2017tg}
Hengfeng Wei, Yu~Huang, and Jian Lu.
\newblock Specification and implementation of replicated list: The {Jupiter}
  protocol revisited.
\newblock {\em arxiv.org}, August 2017.
\newblock URL: \url{https://arxiv.org/abs/1708.04754}.

\bibitem{Weiss:2009ht}
St{\'e}phane Weiss, Pascal Urso, and Pascal Molli.
\newblock {Logoot}: A scalable optimistic replication algorithm for
  collaborative editing on {P2P} networks.
\newblock In {\em 29th IEEE International Conference on Distributed Computing
  Systems (ICDCS)}, pages 404--412, June 2009.
\newblock \href {http://dx.doi.org/10.1109/ICDCS.2009.75}
  {\path{doi:10.1109/ICDCS.2009.75}}.

\bibitem{Weiss:2010hx}
St{\'e}phane Weiss, Pascal Urso, and Pascal Molli.
\newblock {Logoot-Undo}: Distributed collaborative editing system on {P2P}
  networks.
\newblock {\em IEEE Transactions on Parallel and Distributed Systems},
  21(8):1162--1174, January 2010.
\newblock \href {http://dx.doi.org/10.1109/TPDS.2009.173}
  {\path{doi:10.1109/TPDS.2009.173}}.

\bibitem{DBLP:conf/tphol/WenzelPN08}
Makarius Wenzel, Lawrence~C. Paulson, and Tobias Nipkow.
\newblock The {Isabelle} framework.
\newblock In {\em 21st International Conference on Theorem Proving in Higher
  Order Logics (TPHOLs)}, pages 33--38, August 2008.
\newblock \href {http://dx.doi.org/10.1007/978-3-540-71067-7_7}
  {\path{doi:10.1007/978-3-540-71067-7_7}}.

\bibitem{Zeller:2014fl}
Peter Zeller, Annette Bieniusa, and Arnd Poetzsch-Heffter.
\newblock Formal specification and verification of {CRDTs}.
\newblock In {\em 34th IFIP International Conference on Formal Techniques for
  Distributed Objects, Components and Systems (FORTE)}, June 2014.
\newblock \href {http://dx.doi.org/10.1007/978-3-662-43613-4_3}
  {\path{doi:10.1007/978-3-662-43613-4_3}}.

\end{thebibliography}

\ifarxiv
\newpage
\appendix
\section{Generating Operations}\label{sect:appendix:generating-ops}

\floatname{algorithm}{Listing}
\begin{algorithm}
\caption{Generating new operations for modifying maps and lists.}\label{fig:pseudocode}
\noindent
\renewcommand\algorithmicindent{10pt}
\begin{minipage}[t]{0.5\textwidth}
\begin{algorithmic}[0]
    \Function{setMapKey}{$O, \mathit{map}, \mathit{key}, \mathit{val}$}
    \State $(E,\, L) = \llbracket O \rrbracket$
    \State $(\mathit{id}_1, \mathit{op}_1) = \Call{valueID}{O, \mathit{val}}$
    \If{$\mathit{op}_1 \neq \bot$}
    \State $O := O \;\cup\; \big\{ (\mathit{id}_1, \mathit{op}_1) \big\})$
    \EndIf
    \State $\mathit{id}_2 = \mathrm{newID}(O)$
    \State $\mathit{prev} = \{ \mathit{id} \mid \exists\,v.\; (\mathit{id}, \mathit{map}, \mathit{key}, v) \in E \}$
    \State \Return $O \;\cup$
    \State $\quad\big\{ (\mathit{id}_2,\, \mathsf{Assign}(\mathit{map}, \mathit{key}, \mathit{id}_1, \mathit{prev})) \big\}$
    \EndFunction\Statex

    \Function{removeMapKey}{$O, \mathit{map}, \mathit{key}$}
    \State $(E,\, L) = \llbracket O \rrbracket$
    \State $\mathit{id}_1 = \mathrm{newID}(O)$
    \State $\mathit{prev} = \{ \mathit{id} \mid \exists\,v.\; (\mathit{id}, \mathit{map}, \mathit{key}, v) \in E \}$
    \State \Return $O \;\cup$
    \State $\quad\big\{ (\mathit{id}_1,\, \mathsf{Remove}(\mathit{map}, \mathit{key}, \mathit{prev})) \big\}$
    \EndFunction\Statex

    \Function{valueID}{$O, \mathit{val}$}
    \If{$\mathit{val}$ is a primitive type}
    \State \Return $(\mathrm{newID}(O),\, \mathsf{MakeVal}(\mathit{val}))$
    \ElsIf{$\mathit{val} = \texttt{[]}$ (empty list literal)}
    \State \Return $(\mathrm{newID}(O),\, \mathsf{MakeList})$
    \ElsIf{$\mathit{val} = \texttt{{\char '173}{\char '175}}$ (empty map literal)}
    \State \Return $(\mathrm{newID}(O),\, \mathsf{MakeMap})$
    \Else ~($\mathit{val}$ is an existing object)
    \State \Return $(\mathrm{objID}(\mathit{val}),\, \bot)$
    \EndIf
    \EndFunction
\end{algorithmic}
\end{minipage}%
\begin{minipage}[t]{0.5\textwidth}
\begin{algorithmic}[0]
    \Function{setListIndex}{$O, \mathit{list}, \mathit{index}, \mathit{val}$}
    \State $(E,\, L) = \llbracket O \rrbracket$
    \State $(\mathit{id}_1, \mathit{op}_1) = \Call{valueID}{O, \mathit{val}}$
    \If{$\mathit{op}_1 \neq \bot$}
    \State $O := O \;\cup\; \big\{ (\mathit{id}_1, \mathit{op}_1) \big\})$
    \EndIf
    \State $\mathit{id}_2 = \mathrm{newID}(O)$
    \State $\mathit{key} = \mathrm{idxKey}_{E,\, L}(\mathit{list}, \mathit{list}, \mathit{index})$
    \State $\mathit{prev} = \{ \mathit{id} \mid \exists\,v.\; (\mathit{id}, \mathit{list}, \mathit{key}, v) \in E \}$
    \State \Return $O \;\cup$
    \State $\quad\big\{ (\mathit{id}_2,\, \mathsf{Assign}(\mathit{list}, \mathit{key}, \mathit{id}_1, \mathit{prev})) \big\}$
    \EndFunction\Statex

    \Function{insListIndex}{$O, \mathit{list}, \mathit{index}, \mathit{val}$}
    \State $(E,\, L) = \llbracket O \rrbracket$
    \State $(\mathit{id}_1, \mathit{op}_1) = \Call{valueID}{O, \mathit{val}}$
    \If{$\mathit{op}_1 \neq \bot$}
    \State $O := O \;\cup\; \big\{ (\mathit{id}_1, \mathit{op}_1) \big\})$
    \EndIf
    \State $\mathit{id}_2 = \mathrm{newID}(O)$
    \If{$\mathit{index}=0$}
    \State $\mathit{ref} = \mathit{list}$
    \Else
    \State $\mathit{ref} = \mathrm{idxKey}_{E,\, L}(\mathit{list}, \mathit{list}, \mathit{index} - 1)$
    \EndIf
    \State $O := O \;\cup\; \big\{ (\mathit{id}_2,\, \mathsf{InsertAfter}(\mathit{ref})) \big\}$
    \State $\mathit{id}_3 = \mathrm{newID}(O)$
    \State \Return $O \;\cup$
    \State $\quad\big\{ (\mathit{id}_3,\, \mathsf{Assign}(\mathit{list}, \mathit{id}_2, \mathit{id}_1, \emptyset)) \big\}$
    \EndFunction\Statex

    \Function{removeListIndex}{$O, \mathit{list}, \mathit{index}$}
    \State $(E,\, L) = \llbracket O \rrbracket$
    \State $\mathit{id}_1 = \mathrm{newID}(O)$
    \State $\mathit{key} = \mathrm{idxKey}_{E,\, L}(\mathit{list}, \mathit{list}, \mathit{index})$
    \State $\mathit{prev} = \{ \mathit{id} \mid \exists\,v.\; (\mathit{id}, \mathit{list}, \mathit{key}, v) \in E \}$
    \State \Return $O \;\cup$
    \State $\quad\big\{ (\mathit{id}_1,\, \mathsf{Remove}(\mathit{list}, \mathit{key}, \mathit{prev})) \big\}$
    \EndFunction
\end{algorithmic}
\end{minipage}
\end{algorithm}

\noindent
Listing~\ref{fig:pseudocode} gives pseudocode for functions that generate these operations.
Intuitively, these functions form an API through which nodes can modify OpSets to indirectly describe replicated list and map datatypes.
The first parameter $O$ of each function is the OpSet that defines the current state of the node, and the functions return an updated OpSet containing new operations.
The interpretation $\llbracket O \rrbracket$ returns a pair $(E,\, L)$ as defined in \S~\ref{sec:datatypes-interp}.
The function $\mathrm{newID}(O)$ returns a unique ID (e.g.\ a Lamport timestamp \cite{Lamport:1978jq}) that is greater than any existing ID in the OpSet $O$.
Note that we elide explicit network broadcasts reflecting changes in a node's OpSet.

The function \textsc{setMapKey} can be called by a user to update a map object with ID $\mathit{map}$, setting a key $\mathit{key}$ to a value $\mathit{val}$.
If $\mathit{val}$ references an existing object, $\mathit{id}_1$ is set to the ID of that existing object; otherwise, the function \textsc{valueID} generates a new $\mathsf{Make}\cdots$ operation for the value, and the new operation is added to $O'$.
A new ID $\mathit{id}_2$ is generated for the $\mathsf{Assign}$ operation, in which the key $\mathit{key}$ is set to $\mathit{id}_1$ (which identifies the value $\mathit{val}$).
Finally, the function returns the OpSet with the $\mathsf{Assign}$ operation included.
The other definitions follow a similar pattern.

The list manipulation functions \textsc{setListIndex}, \textsc{insListIndex} and \textsc{removeListIndex} take a numeric index as argument to identify the position in the list being edited.
The numeric index is translated into the ID of a list element using the function $\mathrm{idxKey}_{E,\, L}()$:
\[ \mathrm{idxKey}_{E,\, L}(\mathit{obj}, \mathit{key}, i) \;=\; \left\{
   \arraycolsep=2pt \def\arraystretch{1.3}
   \begin{array}{llllll}
       \mathrm{idxKey}_{E,\, L}(\mathit{obj}, n, i-1) \\
       \quad\text{if }\; i > 0 \wedge (\mathit{key}, n) \in L & \wedge &
       \exists\,\mathit{id}, \mathit{val}.\; (\mathit{id}, \mathit{obj}, \mathit{key}, \mathit{val}) \in E \\
       \mathrm{idxKey}_{E,\, L}(\mathit{obj}, n, i) \\
       \quad \hphantom{i > 0 \;\wedge\;} \text{if }\; (\mathit{key}, n) \in L & \wedge &
       \nexists\,\mathit{id}, \mathit{val}.\; (\mathit{id}, \mathit{obj}, \mathit{key}, \mathit{val}) \in E \\
       \mathit{key} \hspace{29.5pt}\text{if }\; i = 0 & \wedge &
       \exists\,\mathit{id}, \mathit{val}.\; (\mathit{id}, \mathit{obj}, \mathit{key}, \mathit{val}) \in E \\
   \end{array} \right. \]
$\mathit{key}$ is initially the ID of the $\mathsf{MakeList}$ operation that created the list.
The function recursively moves along the linked list structure in $L$, decrementing the index for every list element that has an associated value, and not counting any list elements without associated value (which are treated as deleted).
Eventually, it returns the ID of the list element with the desired index.

\section{Introduction to Isabelle/HOL}
\label{sect:appendix:isabelle}

To help any readers who are not familiar with Isabelle/HOL, this appendix provides a brief introduction to the key concepts and syntax, taken from our previous work \cite{Gomes:2017gy}.
A more detailed introduction can be found in the standard tutorial material~\cite{DBLP:books/sp/NipkowK14}.

\subsection{Syntax of expressions.}

Isabelle/HOL is a logic with a strict, polymorphic, inferred type system.
\emph{Function types} are written $\tau_1 \Rightarrow \tau_2$, and are inhabited by \emph{total} functions, mapping elements of $\tau_1$ to elements of $\tau_2$.
We write $\tau_1 \times \tau_2$ for the \emph{product type} of $\tau_1$ and $\tau_2$, inhabited by pairs of elements of type $\tau_1$ and $\tau_2$, respectively.
\emph{Type operators} are applied to arguments in reverse order: $\tau\ \isa{list}$ denotes the type of lists of elements of type $\tau$, and $\tau\ \isa{set}$ denotes the type of mathematical (i.e., potentially infinite) sets of type $\tau$, for instance.
Type variables are written in lowercase, and preceded with a prime: ${\isacharprime}a \Rightarrow {\isacharprime}a$ denotes the type of a polymorphic identity function, for example.
\emph{Tagged union} types are introduced with the $\isacommand{datatype}$ keyword, with \emph{constructors} of these types usually written with an initial upper case letter.

In Isabelle/HOL's term language we write $\isa{t} \mathbin{::} \tau$ for a \emph{type ascription}, constraining the type of the term $\isa{t}$ to the type $\tau$.
We write $\lambda{x}.\: t$ for an anonymous function mapping an argument $\isa{x}$ to $\isa{t(x)}$, and write the application of term $\isa{t}$ with function type to an argument $\isa{u}$ as $\isa{t\ u}$.
Terms of list type are introduced using one of two constructors: the empty list $[\,]$ or `nil', and the infix `cons' operator $\isa{\#}$, which prepends an element to an existing list.
We use $[t_1, \ldots, t_n]$ as syntactic sugar for a list literal, and $\isa{xs} \mathbin{\isacharat} \isa{ys}$ to express the concatenation (appending) of two lists $\isa{xs}$ and $\isa{ys}$.
We write $\{\,\}$ for the empty set, and use usual mathematical notation for set union, disjunction, membership tests, and so on: $\isa{t} \cup \isa{u}$, $\isa{t} \cap \isa{u}$, and $\isa{x} \in \isa{t}$.
We write $t \longrightarrow s$ for logical implication between formulae (terms of type $\isa{bool}$).
Strictly speaking Isabelle is a logical framework, providing a weak meta-logic within which object logics are embedded, including the Isabelle/HOL object logic that we use in this work.
Accordingly, the implication arrow of Isabelle's meta-logic, $\isa{t} \Longrightarrow \isa{u}$, is required in certain contexts over the object-logic implication arrow, $t \longrightarrow s$, already introduced.
However, for purposes of an intuitive understanding, the two forms of implication can be regarded as equivalent by the reader, with the requirement to use one over the other merely being an implementation detail of Isabelle itself.
We will sometimes use the shorthand ${\isasymlbrakk}\isa{H}_1{\isacharsemicolon}\ \ldots{\isacharsemicolon}\ \isa{H}_n{\isasymrbrakk}\ {\isasymLongrightarrow}\ C$ instead of iterated meta-logic implications, i.e., $H_1\ {\isasymLongrightarrow}\ \ldots\ {\isasymLongrightarrow}\ H_n\ {\isasymLongrightarrow}\ C$.

\subsection{Definitions and theorems.}

New non-recursive definitions are entered into Isabelle's global context using the $\mathbf{definition}$ keyword.
Recursive functions are defined using the $\mathbf{fun}$ keyword, and support \emph{pattern matching} on their arguments.
All functions are total, and therefore every recursive function must be provably terminating.
All termination proofs in this work are generated automatically by Isabelle itself.

\emph{Inductive relations} are defined with the $\mathbf{inductive}$ keyword.
For example, the definition
\begin{isabelle}
\isacommand{inductive} only-fives\ {\isacharcolon}{\isacharcolon}\ {\isachardoublequoteopen}nat\ list\ {\isasymRightarrow}\ bool{\isachardoublequoteclose}\ \isakeyword{where} \\
~~~~{\isachardoublequoteopen}only-fives\ {\isacharbrackleft}{\isacharbrackright}{\isachardoublequoteclose}\ {\isacharbar}\\
~~~~{\isachardoublequoteopen}{\isasymlbrakk}\ only-fives\ xs\ {\isasymrbrakk}\ {\isasymLongrightarrow}\ only-fives {\isacharparenleft}5\#xs{\isacharparenright}{\isachardoublequoteclose}
\end{isabelle}
\noindent 
introduces a new constant $\isa{only-fives}$ of type $\isa{nat list} \Rightarrow \isa{bool}$.
The two clauses in the body of the definition enumerate the conditions under which $\isa{only-fives}\ \isa{xs}$ is true, for arbitrary $\isa{xs}$: firstly, $\isa{only-fives}$ is true for the empty list; and secondly, if you know that $\isa{only-fives}\ \isa{xs}$ is true for some $\isa{xs}$, then you can deduce that $\isa{only-fives}\ (5\#\isa{xs})$ (i.e., $\isa{xs}$ prefixed with the number 5) is also true.
Moreover, $\isa{only-fives}\ \isa{xs}$ is true in no other circumstances---it is the \emph{smallest} relation closed under the rules defining it.
In short, the clauses above state that $\isa{only-fives}\ \isa{xs}$ holds exactly in the case where $\isa{xs}$ is a (potentially empty) list containing only repeated copies of the natural number $5$.

Lemmas, theorems, and corollaries can be asserted using the $\isacommand{lemma}$, $\isacommand{theorem}$, and $\isacommand{corollary}$ keywords, respectively.
There is no semantic difference between these keywords in Isabelle, and they serve only to mark certain results as especially important (or unimportant) for human readers.
For example,
\begin{isabelle}
\isacommand{theorem} only-fives-concat{\isacharcolon} \\
~~~~\isakeyword{assumes}\ only-fives\ xs \isakeyword{and}\ only-fives\ ys\\
~~~~\isakeyword{shows}\ only-fives\ (xs \isacharat ys)
\end{isabelle}
\noindent 
conjectures that if $\isa{xs}$ and $\isa{ys}$ are both lists of fives, then their concatenation $xs \mathbin{\isacharat} ys$ is also a list of fives.
Isabelle then requires that this claim be proved by using one of its proof methods, for example by induction.
Some proofs can be automated, whilst others require the user to provide explicit reasoning steps.
The theorem is assigned a name, here $\isa{only-fives-concat}$, so that it may be referenced in later proofs.

\section{Statements of Mechanised Proofs}
\label{sect:appendix:statements}

In this appendix we provide a copy of the Isabelle/HOL definitions and proof statements that support the central claims in the paper.
For space reasons, the actual proofs are omitted; the full formal proof development can be found in the Isabelle Archive of Formal Proofs \cite{AFP}.
The source code is available at \url{https://github.com/trvedata/opsets}.

\subsection{Abstract OpSet}\label{sec:abstract-opset}

In this section, we define a general-purpose OpSet abstraction that is not specific to any one particular datatype.
An OpSet is a set of (ID, operation) pairs with an associated total order on IDs (represented here with the \isa{linorder} typeclass), and satisfying the following properties:
\begin{enumerate}
\item The ID is unique (that is, if any two pairs in the set have the same ID, then their operation is also the same).
\item If the operation references the IDs of any other operations, those referenced IDs are less than that of the operation itself, according to the total order on IDs.
To avoid assuming anything about the structure of operations here, we use a function \isa{deps} that returns the set of dependent IDs for a given operation.
This requirement is a weak expression of causality: an operation can only depend on causally prior operations, and by making the total order on IDs a linear extension of the causal order, we can easily ensure that any referenced IDs are less than that of the operation itself.
\item The OpSet is finite (but we do not assume any particular maximum size).
\end{enumerate}
We define it as follows in Isabelle:\footnote{In programming terms, a \emph{locale} (or 'local theory') may be thought of as an interface with associated laws that implementations must obey.
When showing that an implementation matches this interface, one must also show that the implementation satisfies all assumed laws of the locale.
Moreover, locales can be extended with new assumed facts and fixed constants to form a hierarchy, and definitions and theorems may be defined and declared within a locale and made available to all of its implementations.
See the standard Isabelle tutorial material, as well as~\cite{DBLP:conf/tphol/KammullerWP99} and~\cite{DBLP:conf/types/HaftmannW08} for a more detailed explanation of locales.}
\begin{isabelle}
\isacommand{locale}\isamarkupfalse%
\ opset\ {\isacharequal}\isanewline
\ \ \isakeyword{fixes}\ opset\ {\isacharcolon}{\isacharcolon}\ {\isachardoublequoteopen}{\isacharparenleft}{\isacharprime}oid{\isacharcolon}{\isacharcolon}{\isacharbraceleft}linorder{\isacharbraceright}\ {\isasymtimes}\ {\isacharprime}oper{\isacharparenright}\ set{\isachardoublequoteclose}\isanewline
\ \ \ \ \isakeyword{and}\ deps\ \ {\isacharcolon}{\isacharcolon}\ {\isachardoublequoteopen}{\isacharprime}oper\ {\isasymRightarrow}\ {\isacharprime}oid\ set{\isachardoublequoteclose}\isanewline
\ \ \isakeyword{assumes}\ unique{\isacharunderscore}oid{\isacharcolon}\ {\isachardoublequoteopen}{\isacharparenleft}oid{\isacharcomma}\ op{\isadigit{1}}{\isacharparenright}\ {\isasymin}\ opset\ {\isasymLongrightarrow}\ {\isacharparenleft}oid{\isacharcomma}\ op{\isadigit{2}}{\isacharparenright}\ {\isasymin}\ opset\ {\isasymLongrightarrow}\ op{\isadigit{1}}\ {\isacharequal}\ op{\isadigit{2}}{\isachardoublequoteclose}\isanewline
\ \ \ \ \isakeyword{and}\ ref{\isacharunderscore}older{\isacharcolon}\ {\isachardoublequoteopen}{\isacharparenleft}oid{\isacharcomma}\ oper{\isacharparenright}\ {\isasymin}\ opset\ {\isasymLongrightarrow}\ ref\ {\isasymin}\ deps\ oper\ {\isasymLongrightarrow}\ ref\ {\isacharless}\ oid{\isachardoublequoteclose}\isanewline
\ \ \ \ \isakeyword{and}\ finite{\isacharunderscore}opset{\isacharcolon}\ {\isachardoublequoteopen}finite\ opset{\isachardoublequoteclose}%
\end{isabelle}

We prove that any subset of an OpSet is also a valid OpSet.
This is the case because, although an operation can depend on causally prior operations, the OpSet does not require those prior operations to actually exist.
This weak assumption makes the OpSet model more general and simplifies reasoning about OpSets.
\begin{isabelle}
\isacommand{lemma}\isamarkupfalse%
\ opset{\isacharunderscore}subset{\isacharcolon}\isanewline
\ \ \isakeyword{assumes}\ {\isachardoublequoteopen}opset\ Y\ deps{\isachardoublequoteclose}\isanewline
\ \ \ \ \isakeyword{and}\ {\isachardoublequoteopen}X\ {\isasymsubseteq}\ Y{\isachardoublequoteclose}\isanewline
\ \ \isakeyword{shows}\ {\isachardoublequoteopen}opset\ X\ deps{\isachardoublequoteclose}
\end{isabelle}

\subsubsection{The \isa{spec-ops} predicate}

The \isa{spec-ops} predicate describes a list of (ID, operation) pairs that corresponds to the linearisation of an OpSet, and which we use for sequentially interpreting the OpSet.
A list satisfies \isa{spec-ops} iff it is sorted in ascending order of IDs, if the IDs are unique, and if every operation's dependencies have lower IDs than the operation itself.
A list is implicitly finite in Isabelle/HOL.

\begin{isabelle}
\isacommand{definition}\isamarkupfalse%
\ spec{\isacharunderscore}ops\ {\isacharcolon}{\isacharcolon}\ {\isachardoublequoteopen}{\isacharparenleft}{\isacharprime}oid{\isacharcolon}{\isacharcolon}{\isacharbraceleft}linorder{\isacharbraceright}\ {\isasymtimes}\ {\isacharprime}oper{\isacharparenright}\ list\ {\isasymRightarrow}\ {\isacharparenleft}{\isacharprime}oper\ {\isasymRightarrow}\ {\isacharprime}oid\ set{\isacharparenright}\ {\isasymRightarrow}\ bool{\isachardoublequoteclose}\isanewline
\isakeyword{where}\isanewline
\ \ {\isachardoublequoteopen}spec{\isacharunderscore}ops\ ops\ deps\ {\isasymequiv}\ {\isacharparenleft}sorted\ {\isacharparenleft}map\ fst\ ops{\isacharparenright}\ {\isasymand}\ distinct\ {\isacharparenleft}map\ fst\ ops{\isacharparenright}\ {\isasymand}\isanewline
\ \ \ \ \ \ \ \ \ \ \ {\isacharparenleft}{\isasymforall}oid\ oper\ ref{\isachardot}\ {\isacharparenleft}oid{\isacharcomma}\ oper{\isacharparenright}\ {\isasymin}\ set\ ops\ {\isasymand}\ ref\ {\isasymin}\ deps\ oper\ {\isasymlongrightarrow}\ ref\ {\isacharless}\ oid{\isacharparenright}{\isacharparenright}{\isachardoublequoteclose}
\end{isabelle}
\noindent We prove that for any given OpSet, a \isa{spec-ops} linearisation exists:
\begin{isabelle}
\isacommand{lemma}\isamarkupfalse%
\ spec{\isacharunderscore}ops{\isacharunderscore}exists{\isacharcolon}\isanewline
\ \ \isakeyword{assumes}\ {\isachardoublequoteopen}opset\ ops\ deps{\isachardoublequoteclose}\isanewline
\ \ \isakeyword{shows}\ {\isachardoublequoteopen}{\isasymexists}op{\isacharunderscore}list{\isachardot}\ set\ op{\isacharunderscore}list\ {\isacharequal}\ ops\ {\isasymand}\ spec{\isacharunderscore}ops\ op{\isacharunderscore}list\ deps{\isachardoublequoteclose}
\end{isabelle}
\noindent Conversely, for any given \isa{spec-ops} list, the set of pairs in the list is an OpSet:
\begin{isabelle}
\isacommand{lemma}\isamarkupfalse%
\ spec{\isacharunderscore}ops{\isacharunderscore}is{\isacharunderscore}opset{\isacharcolon}\isanewline
\ \ \isakeyword{assumes}\ {\isachardoublequoteopen}spec{\isacharunderscore}ops\ op{\isacharunderscore}list\ deps{\isachardoublequoteclose}\isanewline
\ \ \isakeyword{shows}\ {\isachardoublequoteopen}opset\ {\isacharparenleft}set\ op{\isacharunderscore}list{\isacharparenright}\ deps{\isachardoublequoteclose}
\end{isabelle}

\subsubsection{The \isa{crdt-ops} predicate}

Like \isa{spec-ops}, the \isa{crdt-ops} predicate describes the linearisation of an OpSet into a list.
Like \isa{spec-ops}, it requires IDs to be unique.
However, its other properties are different: \isa{crdt-ops} does not require operations to appear in sorted order, but instead, whenever any operation references the ID of a prior operation, that prior operation must appear previously in the \isa{crdt-ops} list.
Thus, the order of operations is partially constrained: operations must appear in causal order, but concurrent operations can be ordered arbitrarily.

This list describes the operation sequence in the order it is typically applied to an operation-based CRDT.
Applying operations in the order they appear in \isa{crdt-ops} requires that concurrent operations commute.
For any \isa{crdt-ops} operation sequence, there is a permutation that satisfies the \isa{spec-ops} predicate.
Thus, to check whether a CRDT satisfies its sequential specification, we can prove that interpreting any \isa{crdt-ops} operation sequence with the commutative operation interpretation results in the same end result as interpreting the \isa{spec-ops} permutation of that operation sequence with the sequential operation interpretation.

\begin{isabelle}
\isacommand{inductive}\isamarkupfalse%
\ crdt{\isacharunderscore}ops\ {\isacharcolon}{\isacharcolon}\ {\isachardoublequoteopen}{\isacharparenleft}{\isacharprime}oid{\isacharcolon}{\isacharcolon}{\isacharbraceleft}linorder{\isacharbraceright}\ {\isasymtimes}\ {\isacharprime}oper{\isacharparenright}\ list\ {\isasymRightarrow}\ {\isacharparenleft}{\isacharprime}oper\ {\isasymRightarrow}\ {\isacharprime}oid\ set{\isacharparenright}\ {\isasymRightarrow}\ bool{\isachardoublequoteclose}\isanewline
\isakeyword{where}\isanewline
\ \ {\isachardoublequoteopen}crdt{\isacharunderscore}ops\ {\isacharbrackleft}{\isacharbrackright}\ deps{\isachardoublequoteclose}\ {\isacharbar}\isanewline
\ \ {\isachardoublequoteopen}{\isasymlbrakk}crdt{\isacharunderscore}ops\ xs\ deps{\isacharsemicolon}\isanewline
\ \ \ \ oid\ {\isasymnotin}\ set\ {\isacharparenleft}map\ fst\ xs{\isacharparenright}{\isacharsemicolon}\isanewline
\ \ \ \ {\isasymforall}ref\ {\isasymin}\ deps\ oper{\isachardot}\ ref\ {\isasymin}\ set\ {\isacharparenleft}map\ fst\ xs{\isacharparenright}\ {\isasymand}\ ref\ {\isacharless}\ oid\isanewline
\ \ \ {\isasymrbrakk}\ {\isasymLongrightarrow}\ crdt{\isacharunderscore}ops\ {\isacharparenleft}xs\ {\isacharat}\ {\isacharbrackleft}{\isacharparenleft}oid{\isacharcomma}\ oper{\isacharparenright}{\isacharbrackright}{\isacharparenright}\ deps{\isachardoublequoteclose}
\end{isabelle}

\subsection{Specifying List Insertion}

In this section we consider only list insertion.
We model an insertion operation as a pair (\isa{ID, ref}), where \isa{ref} is either \isa{None} (signifying an insertion at the head of the list) or \isa{Some r} (an insertion immediately after a reference element with ID \isa{r}).
If the reference element does not exist, the operation does nothing.

We provide two different definitions of the interpretation function for list insertion: \isa{insert-spec} and \isa{insert-alt}.
The \isa{insert-alt} definition matches the paper, while \isa{insert-spec} uses the Isabelle/HOL list datatype, making it more suitable for formal reasoning.
In section~\ref{sec:insert-alt-equiv} we prove that the two definitions are in fact equivalent.

\begin{isabelle}
\isacommand{fun}\isamarkupfalse%
\ insert{\isacharunderscore}spec\ {\isacharcolon}{\isacharcolon}\ {\isachardoublequoteopen}{\isacharprime}oid\ list\ {\isasymRightarrow}\ {\isacharparenleft}{\isacharprime}oid\ {\isasymtimes}\ {\isacharprime}oid\ option{\isacharparenright}\ {\isasymRightarrow}\ {\isacharprime}oid\ list{\isachardoublequoteclose}\isanewline
\isakeyword{where}\isanewline
\ \ {\isachardoublequoteopen}insert{\isacharunderscore}spec\ xs\ \ \ \ \ {\isacharparenleft}oid{\isacharcomma}\ None{\isacharparenright}\ \ \ \ \ {\isacharequal}\ oid{\isacharhash}xs{\isachardoublequoteclose}\ {\isacharbar}\isanewline
\ \ {\isachardoublequoteopen}insert{\isacharunderscore}spec\ {\isacharbrackleft}{\isacharbrackright}\ \ \ \ \ {\isacharparenleft}oid{\isacharcomma}\ {\isacharunderscore}{\isacharparenright}\ \ \ \ \ \ \ \ {\isacharequal}\ {\isacharbrackleft}{\isacharbrackright}{\isachardoublequoteclose}\ {\isacharbar}\isanewline
\ \ {\isachardoublequoteopen}insert{\isacharunderscore}spec\ {\isacharparenleft}x{\isacharhash}xs{\isacharparenright}\ {\isacharparenleft}oid{\isacharcomma}\ Some\ ref{\isacharparenright}\ {\isacharequal}\isanewline
\ \ \ \ \ {\isacharparenleft}if\ x\ {\isacharequal}\ ref\ then\ x\ {\isacharhash}\ oid\ {\isacharhash}\ xs\isanewline
\ \ \ \ \ \ \ \ \ \ \ \ \ \ \ \ \ else\ x\ {\isacharhash}\ {\isacharparenleft}insert{\isacharunderscore}spec\ xs\ {\isacharparenleft}oid{\isacharcomma}\ Some\ ref{\isacharparenright}{\isacharparenright}{\isacharparenright}{\isachardoublequoteclose}\isanewline
\isanewline
\isacommand{fun}\isamarkupfalse%
\ insert{\isacharunderscore}alt\ {\isacharcolon}{\isacharcolon}\ {\isachardoublequoteopen}{\isacharparenleft}{\isacharprime}oid\ {\isasymtimes}\ {\isacharprime}oid\ option{\isacharparenright}\ set\ {\isasymRightarrow}\ {\isacharparenleft}{\isacharprime}oid\ {\isasymtimes}\ {\isacharprime}oid{\isacharparenright}\ {\isasymRightarrow}\ {\isacharparenleft}{\isacharprime}oid\ {\isasymtimes}\ {\isacharprime}oid\ option{\isacharparenright}\ set{\isachardoublequoteclose}\isanewline
\isakeyword{where}\isanewline
\ \ {\isachardoublequoteopen}insert{\isacharunderscore}alt\ list{\isacharunderscore}rel\ {\isacharparenleft}oid{\isacharcomma}\ ref{\isacharparenright}\ {\isacharequal}\ {\isacharparenleft}\isanewline
\ \ \ \ \ \ if\ {\isasymexists}n{\isachardot}\ {\isacharparenleft}ref{\isacharcomma}\ n{\isacharparenright}\ {\isasymin}\ list{\isacharunderscore}rel\isanewline
\ \ \ \ \ \ then\ {\isacharbraceleft}{\isacharparenleft}p{\isacharcomma}\ n{\isacharparenright}\ {\isasymin}\ list{\isacharunderscore}rel{\isachardot}\ p\ {\isasymnoteq}\ ref{\isacharbraceright}\ {\isasymunion}\ {\isacharbraceleft}{\isacharparenleft}ref{\isacharcomma}\ Some\ oid{\isacharparenright}{\isacharbraceright}\ {\isasymunion}\isanewline
\ \ \ \ \ \ \ \ \ \ \ {\isacharbraceleft}{\isacharparenleft}i{\isacharcomma}\ n{\isacharparenright}{\isachardot}\ i\ {\isacharequal}\ oid\ {\isasymand}\ {\isacharparenleft}ref{\isacharcomma}\ n{\isacharparenright}\ {\isasymin}\ list{\isacharunderscore}rel{\isacharbraceright}\isanewline
\ \ \ \ \ \ else\ list{\isacharunderscore}rel{\isacharparenright}{\isachardoublequoteclose}%
\end{isabelle}

\isa{interp-ins} is the sequential interpretation of a set of insertion operations.
It starts with an empty list as initial state, and then applies the operations from left to right.

\begin{isabelle}
\isacommand{definition}\isamarkupfalse%
\ interp{\isacharunderscore}ins\ {\isacharcolon}{\isacharcolon}\ {\isachardoublequoteopen}{\isacharparenleft}{\isacharprime}oid\ {\isasymtimes}\ {\isacharprime}oid\ option{\isacharparenright}\ list\ {\isasymRightarrow}\ {\isacharprime}oid\ list{\isachardoublequoteclose}\ \isakeyword{where}\isanewline
\ \ {\isachardoublequoteopen}interp{\isacharunderscore}ins\ ops\ {\isasymequiv}\ foldl\ insert{\isacharunderscore}spec\ {\isacharbrackleft}{\isacharbrackright}\ ops{\isachardoublequoteclose}%
\end{isabelle}

\subsubsection{The \isa{insert-ops} predicate}

We now specialise the definitions from section~\ref{sec:abstract-opset} for list insertion.
\isa{insert-opset} is an opset consisting only of insertion operations, and \isa{insert-ops} is the specialisation of the \isa{spec-ops} predicate for insertion operations.

\begin{isabelle}
\isacommand{locale}\isamarkupfalse%
\ insert{\isacharunderscore}opset\ {\isacharequal}\ opset\ opset\ set{\isacharunderscore}option\isanewline
\ \ \isakeyword{for}\ opset\ {\isacharcolon}{\isacharcolon}\ {\isachardoublequoteopen}{\isacharparenleft}{\isacharprime}oid{\isacharcolon}{\isacharcolon}{\isacharbraceleft}linorder{\isacharbraceright}\ {\isasymtimes}\ {\isacharprime}oid\ option{\isacharparenright}\ set{\isachardoublequoteclose}\isanewline
\isanewline
\isacommand{definition}\isamarkupfalse%
\ insert{\isacharunderscore}ops\ {\isacharcolon}{\isacharcolon}\ {\isachardoublequoteopen}{\isacharparenleft}{\isacharprime}oid{\isacharcolon}{\isacharcolon}{\isacharbraceleft}linorder{\isacharbraceright}\ {\isasymtimes}\ {\isacharprime}oid\ option{\isacharparenright}\ list\ {\isasymRightarrow}\ bool{\isachardoublequoteclose}\ \isakeyword{where}\isanewline
\ \ {\isachardoublequoteopen}insert{\isacharunderscore}ops\ list\ {\isasymequiv}\ spec{\isacharunderscore}ops\ list\ set{\isacharunderscore}option{\isachardoublequoteclose}
\end{isabelle}

\subsubsection{Equivalence of the two definitions of insertion}\label{sec:insert-alt-equiv}

We now prove that the two definitions of insertion, \isa{insert-spec} and \isa{insert-alt}, are equivalent.
First we define how to derive the successor relation from an Isabelle list.
This relation contains (\isa{id}, \isa{None}) if \isa{id} is the last element of the list, and (\isa{id1}, \isa{id2}) if \isa{id1} is immediately followed by \isa{id2} in the list.

\begin{isabelle}
\isacommand{fun}\isamarkupfalse%
\ succ{\isacharunderscore}rel\ {\isacharcolon}{\isacharcolon}\ {\isachardoublequoteopen}{\isacharprime}oid\ list\ {\isasymRightarrow}\ {\isacharparenleft}{\isacharprime}oid\ {\isasymtimes}\ {\isacharprime}oid\ option{\isacharparenright}\ set{\isachardoublequoteclose}\isanewline
\isakeyword{where}\isanewline
\ \ {\isachardoublequoteopen}succ{\isacharunderscore}rel\ {\isacharbrackleft}{\isacharbrackright}\ {\isacharequal}\ {\isacharbraceleft}{\isacharbraceright}{\isachardoublequoteclose}\ {\isacharbar}\isanewline
\ \ {\isachardoublequoteopen}succ{\isacharunderscore}rel\ {\isacharbrackleft}head{\isacharbrackright}\ {\isacharequal}\ {\isacharbraceleft}{\isacharparenleft}head{\isacharcomma}\ None{\isacharparenright}{\isacharbraceright}{\isachardoublequoteclose}\ {\isacharbar}\isanewline
\ \ {\isachardoublequoteopen}succ{\isacharunderscore}rel\ {\isacharparenleft}head{\isacharhash}x{\isacharhash}xs{\isacharparenright}\ {\isacharequal}\ {\isacharbraceleft}{\isacharparenleft}head{\isacharcomma}\ Some\ x{\isacharparenright}{\isacharbraceright}\ {\isasymunion}\ succ{\isacharunderscore}rel\ {\isacharparenleft}x{\isacharhash}xs{\isacharparenright}{\isachardoublequoteclose}%
\end{isabelle}

\noindent
interp-alt is the equivalent of interp-ins, but using \isa{insert-alt} instead of \isa{insert-spec}.
To match the paper, it uses a distinct head element to refer to the beginning of the list.

\begin{isabelle}
\isacommand{definition}\isamarkupfalse%
\ interp{\isacharunderscore}alt\ {\isacharcolon}{\isacharcolon}\ {\isachardoublequoteopen}{\isacharprime}oid\ {\isasymRightarrow}\ {\isacharparenleft}{\isacharprime}oid\ {\isasymtimes}\ {\isacharprime}oid\ option{\isacharparenright}\ list\ {\isasymRightarrow}\ {\isacharparenleft}{\isacharprime}oid\ {\isasymtimes}\ {\isacharprime}oid\ option{\isacharparenright}\ set{\isachardoublequoteclose}\isanewline
\isakeyword{where}\isanewline
\ \ {\isachardoublequoteopen}interp{\isacharunderscore}alt\ head\ ops\ {\isasymequiv}\ foldl\ insert{\isacharunderscore}alt\ {\isacharbraceleft}{\isacharparenleft}head{\isacharcomma}\ None{\isacharparenright}{\isacharbraceright}\isanewline
\ \ \ \ \ {\isacharparenleft}map\ {\isacharparenleft}{\isasymlambda}x{\isachardot}\ case\ x\ of\isanewline
\ \ \ \ \ \ \ \ \ \ \ \ {\isacharparenleft}oid{\isacharcomma}\ None{\isacharparenright}\ \ \ \ \ {\isasymRightarrow}\ {\isacharparenleft}oid{\isacharcomma}\ head{\isacharparenright}\ {\isacharbar}\isanewline
\ \ \ \ \ \ \ \ \ \ \ \ {\isacharparenleft}oid{\isacharcomma}\ Some\ ref{\isacharparenright}\ {\isasymRightarrow}\ {\isacharparenleft}oid{\isacharcomma}\ ref{\isacharparenright}{\isacharparenright}\ \isanewline
\ \ \ \ \ \ ops{\isacharparenright}{\isachardoublequoteclose}
\end{isabelle}

\noindent We can now prove that \isa{insert-spec} and \isa{insert-alt} are equivalent:
\begin{isabelle}
\isacommand{theorem}\isamarkupfalse%
\ insert{\isacharunderscore}alt{\isacharunderscore}equivalent{\isacharcolon}\isanewline
\ \ \isakeyword{assumes}\ {\isachardoublequoteopen}insert{\isacharunderscore}ops\ ops{\isachardoublequoteclose}\isanewline
\ \ \ \ \isakeyword{and}\ {\isachardoublequoteopen}head\ {\isasymnotin}\ fst\ {\isacharbackquote}\ set\ ops{\isachardoublequoteclose}\isanewline
\ \ \ \ \isakeyword{and}\ {\isachardoublequoteopen}{\isasymAnd}r{\isachardot}\ Some\ r\ {\isasymin}\ snd\ {\isacharbackquote}\ set\ ops\ {\isasymLongrightarrow}\ r\ {\isasymnoteq}\ head{\isachardoublequoteclose}\isanewline
\ \ \isakeyword{shows}\ {\isachardoublequoteopen}succ{\isacharunderscore}rel\ {\isacharparenleft}head\ {\isacharhash}\ interp{\isacharunderscore}ins\ ops{\isacharparenright}\ {\isacharequal}\ interp{\isacharunderscore}alt\ head\ ops{\isachardoublequoteclose}
\end{isabelle}

\subsection{No Interleaving}\label{appendix:no-interleaving}

The predicate \isa{insert-seq start ops} is true iff \isa{ops} is a list of insertion operations that begins by inserting after \isa{start}, and then continues by placing each subsequent insertion directly after its predecessor.
This definition models the sequential insertion of text at a particular place in a text document.%
\begin{isabelle}
\isacommand{inductive}\isamarkupfalse%
\ insert{\isacharunderscore}seq\ {\isacharcolon}{\isacharcolon}\ {\isachardoublequoteopen}{\isacharprime}oid\ option\ {\isasymRightarrow}\ {\isacharparenleft}{\isacharprime}oid\ {\isasymtimes}\ {\isacharprime}oid\ option{\isacharparenright}\ list\ {\isasymRightarrow}\ bool{\isachardoublequoteclose}\ \isakeyword{where}\isanewline
\ \ {\isachardoublequoteopen}insert{\isacharunderscore}seq\ start\ {\isacharbrackleft}{\isacharparenleft}oid{\isacharcomma}\ start{\isacharparenright}{\isacharbrackright}{\isachardoublequoteclose}\ {\isacharbar}\isanewline
\ \ {\isachardoublequoteopen}{\isasymlbrakk}insert{\isacharunderscore}seq\ start\ {\isacharparenleft}list\ {\isacharat}\ {\isacharbrackleft}{\isacharparenleft}prev{\isacharcomma}\ ref{\isacharparenright}{\isacharbrackright}{\isacharparenright}{\isasymrbrakk}\isanewline
\ \ \ \ \ \ {\isasymLongrightarrow}\ insert{\isacharunderscore}seq\ start\ {\isacharparenleft}list\ {\isacharat}\ {\isacharbrackleft}{\isacharparenleft}prev{\isacharcomma}\ ref{\isacharparenright}{\isacharcomma}\ {\isacharparenleft}oid{\isacharcomma}\ Some\ prev{\isacharparenright}{\isacharbrackright}{\isacharparenright}{\isachardoublequoteclose}
\end{isabelle}

Consider an execution that contains two distinct insertion sequences, \isa{xs} and \isa{ys}, that both begin at the same initial position \isa{start}.
We prove that, provided the starting element exists, the two insertion sequences are not interleaved.
That is, in the final list order, either all insertions by \isa{xs} appear before all insertions by \isa{ys}, or vice versa.

\begin{isabelle}
\isacommand{theorem}\isamarkupfalse%
\ no{\isacharunderscore}interleaving{\isacharcolon}\isanewline
\ \ \isakeyword{assumes}\ {\isachardoublequoteopen}insert{\isacharunderscore}ops\ ops{\isachardoublequoteclose}\isanewline
\ \ \ \ \isakeyword{and}\ {\isachardoublequoteopen}insert{\isacharunderscore}seq\ start\ xs{\isachardoublequoteclose}\ \isakeyword{and}\ {\isachardoublequoteopen}insert{\isacharunderscore}ops\ xs{\isachardoublequoteclose}\isanewline
\ \ \ \ \isakeyword{and}\ {\isachardoublequoteopen}insert{\isacharunderscore}seq\ start\ ys{\isachardoublequoteclose}\ \isakeyword{and}\ {\isachardoublequoteopen}insert{\isacharunderscore}ops\ ys{\isachardoublequoteclose}\isanewline
\ \ \ \ \isakeyword{and}\ {\isachardoublequoteopen}set\ xs\ {\isasymsubseteq}\ set\ ops{\isachardoublequoteclose}\ \isakeyword{and}\ {\isachardoublequoteopen}set\ ys\ {\isasymsubseteq}\ set\ ops{\isachardoublequoteclose}\isanewline
\ \ \ \ \isakeyword{and}\ {\isachardoublequoteopen}distinct\ {\isacharparenleft}map\ fst\ xs\ {\isacharat}\ map\ fst\ ys{\isacharparenright}{\isachardoublequoteclose}\isanewline
\ \ \ \ \isakeyword{and}\ {\isachardoublequoteopen}start\ {\isacharequal}\ None\ {\isasymor}\ {\isacharparenleft}{\isasymexists}r{\isachardot}\ start\ {\isacharequal}\ Some\ r\ {\isasymand}\ r\ {\isasymin}\ set\ {\isacharparenleft}interp{\isacharunderscore}ins\ ops{\isacharparenright}{\isacharparenright}{\isachardoublequoteclose}\isanewline
\ \ \isakeyword{shows}\ {\isachardoublequoteopen}{\isacharparenleft}{\isasymforall}x\ {\isasymin}\ set\ {\isacharparenleft}map\ fst\ xs{\isacharparenright}{\isachardot}\ {\isasymforall}y\ {\isasymin}\ set\ {\isacharparenleft}map\ fst\ ys{\isacharparenright}{\isachardot}\ list{\isacharunderscore}order\ ops\ x\ y{\isacharparenright}\ {\isasymor}\isanewline
\ \ \ \ \ \ \ \ \ {\isacharparenleft}{\isasymforall}x\ {\isasymin}\ set\ {\isacharparenleft}map\ fst\ xs{\isacharparenright}{\isachardot}\ {\isasymforall}y\ {\isasymin}\ set\ {\isacharparenleft}map\ fst\ ys{\isacharparenright}{\isachardot}\ list{\isacharunderscore}order\ ops\ y\ x{\isacharparenright}{\isachardoublequoteclose}
\end{isabelle}

\noindent
For completeness, we also prove what happens if there are two insertion sequences, \isa{xs} and \isa{ys}, but their reference element \isa{start} does not exist.
In this failure case, none of the insertions in \isa{xs} or \isa{ys} take effect.
\begin{isabelle}
\isacommand{theorem}\isamarkupfalse%
\ missing{\isacharunderscore}start{\isacharunderscore}no{\isacharunderscore}insertion{\isacharcolon}\isanewline
\ \ \isakeyword{assumes}\ {\isachardoublequoteopen}insert{\isacharunderscore}ops\ ops{\isachardoublequoteclose}\isanewline
\ \ \ \ \isakeyword{and}\ {\isachardoublequoteopen}insert{\isacharunderscore}seq\ {\isacharparenleft}Some\ start{\isacharparenright}\ xs{\isachardoublequoteclose}\ \isakeyword{and}\ {\isachardoublequoteopen}insert{\isacharunderscore}ops\ xs{\isachardoublequoteclose}\isanewline
\ \ \ \ \isakeyword{and}\ {\isachardoublequoteopen}insert{\isacharunderscore}seq\ {\isacharparenleft}Some\ start{\isacharparenright}\ ys{\isachardoublequoteclose}\ \isakeyword{and}\ {\isachardoublequoteopen}insert{\isacharunderscore}ops\ ys{\isachardoublequoteclose}\isanewline
\ \ \ \ \isakeyword{and}\ {\isachardoublequoteopen}set\ xs\ {\isasymsubseteq}\ set\ ops{\isachardoublequoteclose}\ \isakeyword{and}\ {\isachardoublequoteopen}set\ ys\ {\isasymsubseteq}\ set\ ops{\isachardoublequoteclose}\isanewline
\ \ \ \ \isakeyword{and}\ {\isachardoublequoteopen}start\ {\isasymnotin}\ set\ {\isacharparenleft}interp{\isacharunderscore}ins\ ops{\isacharparenright}{\isachardoublequoteclose}\isanewline
\ \ \isakeyword{shows}\ {\isachardoublequoteopen}{\isasymforall}x\ {\isasymin}\ set\ {\isacharparenleft}map\ fst\ xs{\isacharparenright}\ {\isasymunion}\ set\ {\isacharparenleft}map\ fst\ ys{\isacharparenright}{\isachardot}\ x\ {\isasymnotin}\ set\ {\isacharparenleft}interp{\isacharunderscore}ins\ ops{\isacharparenright}{\isachardoublequoteclose}
\end{isabelle}

\subsection{The Replicated Growable Array (RGA)}\label{appendix:rga}

The RGA algorithm \cite{Roh:2011dw} is a replicated list (or collaborative text-editing) algorithm.
In this section we prove that RGA satisfies our list specification.
The Isabelle/HOL definition of RGA in this section is based on our prior work on formally verifying CRDTs \cite{Gomes:2017gy,Gomes:2017vo}.

\begin{isabelle}
\isacommand{fun}\isamarkupfalse%
\ insert{\isacharunderscore}body\ {\isacharcolon}{\isacharcolon}\ {\isachardoublequoteopen}{\isacharprime}oid{\isacharcolon}{\isacharcolon}{\isacharbraceleft}linorder{\isacharbraceright}\ list\ {\isasymRightarrow}\ {\isacharprime}oid\ {\isasymRightarrow}\ {\isacharprime}oid\ list{\isachardoublequoteclose}\ \isakeyword{where}\isanewline
\ \ {\isachardoublequoteopen}insert{\isacharunderscore}body\ {\isacharbrackleft}{\isacharbrackright}\ \ \ \ \ \ \ e\ {\isacharequal}\ {\isacharbrackleft}e{\isacharbrackright}{\isachardoublequoteclose}\ {\isacharbar}\isanewline
\ \ {\isachardoublequoteopen}insert{\isacharunderscore}body\ {\isacharparenleft}x\ {\isacharhash}\ xs{\isacharparenright}\ e\ {\isacharequal}\isanewline
\ \ \ \ \ {\isacharparenleft}if\ x\ {\isacharless}\ e\ then\ e\ {\isacharhash}\ x\ {\isacharhash}\ xs\isanewline
\ \ \ \ \ \ \ \ \ \ \ \ \ \ \ else\ x\ {\isacharhash}\ insert{\isacharunderscore}body\ xs\ e{\isacharparenright}{\isachardoublequoteclose}\isanewline
\isanewline
\isacommand{fun}\isamarkupfalse%
\ insert{\isacharunderscore}rga\ {\isacharcolon}{\isacharcolon}\ {\isachardoublequoteopen}{\isacharprime}oid{\isacharcolon}{\isacharcolon}{\isacharbraceleft}linorder{\isacharbraceright}\ list\ {\isasymRightarrow}\ {\isacharparenleft}{\isacharprime}oid\ {\isasymtimes}\ {\isacharprime}oid\ option{\isacharparenright}\ {\isasymRightarrow}\ {\isacharprime}oid\ list{\isachardoublequoteclose}\ \isakeyword{where}\isanewline
\ \ {\isachardoublequoteopen}insert{\isacharunderscore}rga\ \ xs\ \ \ \ \ \ {\isacharparenleft}e{\isacharcomma}\ None{\isacharparenright}\ \ \ {\isacharequal}\ insert{\isacharunderscore}body\ xs\ e{\isachardoublequoteclose}\ {\isacharbar}\isanewline
\ \ {\isachardoublequoteopen}insert{\isacharunderscore}rga\ \ {\isacharbrackleft}{\isacharbrackright}\ \ \ \ \ \ {\isacharparenleft}e{\isacharcomma}\ Some\ i{\isacharparenright}\ {\isacharequal}\ {\isacharbrackleft}{\isacharbrackright}{\isachardoublequoteclose}\ {\isacharbar}\isanewline
\ \ {\isachardoublequoteopen}insert{\isacharunderscore}rga\ {\isacharparenleft}x\ {\isacharhash}\ xs{\isacharparenright}\ {\isacharparenleft}e{\isacharcomma}\ Some\ i{\isacharparenright}\ {\isacharequal}\isanewline
\ \ \ \ \ {\isacharparenleft}if\ x\ {\isacharequal}\ i\ then\isanewline
\ \ \ \ \ \ \ \ x\ {\isacharhash}\ insert{\isacharunderscore}body\ xs\ e\isanewline
\ \ \ \ \ \ else\isanewline
\ \ \ \ \ \ \ \ x\ {\isacharhash}\ insert{\isacharunderscore}rga\ xs\ {\isacharparenleft}e{\isacharcomma}\ Some\ i{\isacharparenright}{\isacharparenright}{\isachardoublequoteclose}\isanewline
\isanewline
\isacommand{definition}\isamarkupfalse%
\ interp{\isacharunderscore}rga\ {\isacharcolon}{\isacharcolon}\ {\isachardoublequoteopen}{\isacharparenleft}{\isacharprime}oid{\isacharcolon}{\isacharcolon}{\isacharbraceleft}linorder{\isacharbraceright}\ {\isasymtimes}\ {\isacharprime}oid\ option{\isacharparenright}\ list\ {\isasymRightarrow}\ {\isacharprime}oid\ list{\isachardoublequoteclose}\ \isakeyword{where}\isanewline
\ \ {\isachardoublequoteopen}interp{\isacharunderscore}rga\ ops\ {\isasymequiv}\ foldl\ insert{\isacharunderscore}rga\ {\isacharbrackleft}{\isacharbrackright}\ ops{\isachardoublequoteclose}\isanewline
\isanewline
\isacommand{definition}\isamarkupfalse%
\ rga{\isacharunderscore}ops\ {\isacharcolon}{\isacharcolon}\ {\isachardoublequoteopen}{\isacharparenleft}{\isacharprime}oid{\isacharcolon}{\isacharcolon}{\isacharbraceleft}linorder{\isacharbraceright}\ {\isasymtimes}\ {\isacharprime}oid\ option{\isacharparenright}\ list\ {\isasymRightarrow}\ bool{\isachardoublequoteclose}\ \isakeyword{where}\isanewline
\ \ {\isachardoublequoteopen}rga{\isacharunderscore}ops\ list\ {\isasymequiv}\ crdt{\isacharunderscore}ops\ list\ set{\isacharunderscore}option{\isachardoublequoteclose}
\end{isabelle}

\noindent We can then prove that RGA satisfies our list specification:
\begin{isabelle}
\isacommand{theorem}\isamarkupfalse%
\ rga{\isacharunderscore}meets{\isacharunderscore}spec{\isacharcolon}\isanewline
\ \ \isakeyword{assumes}\ {\isachardoublequoteopen}rga{\isacharunderscore}ops\ xs{\isachardoublequoteclose}\isanewline
\ \ \isakeyword{shows}\ {\isachardoublequoteopen}{\isasymexists}ys{\isachardot}\ set\ ys\ {\isacharequal}\ set\ xs\ {\isasymand}\ insert{\isacharunderscore}ops\ ys\ {\isasymand}\ interp{\isacharunderscore}ins\ ys\ {\isacharequal}\ interp{\isacharunderscore}rga\ xs{\isachardoublequoteclose}
\end{isabelle}

\subsection{Relationship to Strong List Specification}\label{appendix:attiya-spec}

In this section we show that our list specification is stronger than the $\mathcal{A}_\textsf{strong}$ specification of collaborative text editing by Attiya et al.~\cite{Attiya:2016kh}.
We do this by showing that the OpSet interpretation of any set of insertion and deletion operations satisfies all of the consistency criteria that constitute the $\mathcal{A}_\textsf{strong}$ specification.

Attiya et al.'s specification is as follows~\cite{Attiya:2016kh}:

\begin{displayquote}
An abstract execution $A = (H, \textsf{vis})$ belongs to the \emph{strong list specification} $\mathcal{A}_\textsf{strong}$ if and only if there is a relation $\textsf{lo} \subseteq \textsf{elems}(A) \times \textsf{elems}(A)$, called the \emph{list order}, such that:
\begin{enumerate}
\item Each event $e = \mathit{do}(\mathit{op}, w) \in H$ returns a sequence of elements $w=a_0 \dots a_{n-1}$, where $a_i \in \textsf{elems}(A)$, such that
\begin{enumerate}
\item $w$ contains exactly the elements visible to $e$ that have been inserted, but not deleted:
\[ \forall a.\; a \in w \quad\Longleftrightarrow\quad (\mathit{do}(\textsf{ins}(a, \_), \_) \le_\textsf{vis} e)
\;\wedge\; \neg(\mathit{do}(\textsf{del}(a), \_) \le_\textsf{vis} e). \]
\item The order of the elements is consistent with the list order:
\[ \forall i, j.\; (i < j) \;\Longrightarrow\; (a_i, a_j) \in \textsf{lo}. \]
\item Elements are inserted at the specified position:
if $\mathit{op} = \textsf{ins}(a, k)$, then $a = a_{\mathrm{min} \{k,\; n-1\}}$.
\end{enumerate}
\item The list order $\textsf{lo}$ is transitive, irreflexive and total, and thus determines the order of all insert operations in the execution.
\end{enumerate}
\end{displayquote}

This specification considers only insertion and deletion operations, but no assignment.
Moreover, it considers only a single list object, not a graph of composable objects like in our paper.
Thus, we prove the relationship to $\mathcal{A}_\textsf{strong}$ using a simplified interpretation function that defines only insertion and deletion on a single list.

We first define a datatype for list operations, with two constructors: \isa{Insert ref val}, and \isa{Delete ref}.
For insertion, the \isa{ref} argument is the ID of the existing element after which we want to insert, or \isa{None} to insert at the head of the list.
The \isa{val} argument is an arbitrary value to associate with the list element.
For deletion, the \isa{ref} argument is the ID of the existing list element to delete.

\begin{isabelle}
\isacommand{datatype}\isamarkupfalse%
\ {\isacharparenleft}{\isacharprime}oid{\isacharcomma}\ {\isacharprime}val{\isacharparenright}\ list{\isacharunderscore}op\ {\isacharequal}\isanewline
\ \ Insert\ {\isachardoublequoteopen}{\isacharprime}oid\ option{\isachardoublequoteclose}\ {\isachardoublequoteopen}{\isacharprime}val{\isachardoublequoteclose}\ {\isacharbar}\isanewline
\ \ Delete\ {\isachardoublequoteopen}{\isacharprime}oid{\isachardoublequoteclose}%
\end{isabelle}

When interpreting operations, the result is a pair (\isa{list, vals}).
The \isa{list} contains the IDs of list elements in the correct order (equivalent to the list relation in the paper), and \isa{vals} is a mapping from list element IDs to values (equivalent to the element relation in the paper).

Insertion delegates to the previously defined \isa{insert-spec} interpretation function.
Deleting a list element removes it from \isa{vals}.

\begin{isabelle}
\isacommand{fun}\isamarkupfalse%
\ interp{\isacharunderscore}op\ {\isacharcolon}{\isacharcolon}\ {\isachardoublequoteopen}{\isacharparenleft}{\isacharprime}oid\ list\ {\isasymtimes}\ {\isacharparenleft}{\isacharprime}oid\ {\isasymrightharpoonup}\ {\isacharprime}val{\isacharparenright}{\isacharparenright}\ {\isasymRightarrow}\ {\isacharparenleft}{\isacharprime}oid\ {\isasymtimes}\ {\isacharparenleft}{\isacharprime}oid{\isacharcomma}\ {\isacharprime}val{\isacharparenright}\ list{\isacharunderscore}op{\isacharparenright}\isanewline
\ \ \ \ \ \ \ \ \ \ \ \ \ \ \ {\isasymRightarrow}\ {\isacharparenleft}{\isacharprime}oid\ list\ {\isasymtimes}\ {\isacharparenleft}{\isacharprime}oid\ {\isasymrightharpoonup}\ {\isacharprime}val{\isacharparenright}{\isacharparenright}{\isachardoublequoteclose}\ \isakeyword{where}\isanewline
\ \ {\isachardoublequoteopen}interp{\isacharunderscore}op\ {\isacharparenleft}list{\isacharcomma}\ vals{\isacharparenright}\ {\isacharparenleft}oid{\isacharcomma}\ Insert\ ref\ val{\isacharparenright}\ {\isacharequal}\ {\isacharparenleft}insert{\isacharunderscore}spec\ list\ {\isacharparenleft}oid{\isacharcomma}\ ref{\isacharparenright}{\isacharcomma}\ vals{\isacharparenleft}oid\ {\isasymmapsto}\ val{\isacharparenright}{\isacharparenright}{\isachardoublequoteclose}\ {\isacharbar}\isanewline
\ \ {\isachardoublequoteopen}interp{\isacharunderscore}op\ {\isacharparenleft}list{\isacharcomma}\ vals{\isacharparenright}\ {\isacharparenleft}oid{\isacharcomma}\ Delete\ ref\ \ \ \ {\isacharparenright}\ {\isacharequal}\ {\isacharparenleft}list{\isacharcomma}\ vals{\isacharparenleft}ref\ {\isacharcolon}{\isacharequal}\ None{\isacharparenright}{\isacharparenright}{\isachardoublequoteclose}\isanewline
\isanewline
\isacommand{definition}\isamarkupfalse%
\ interp{\isacharunderscore}ops\ {\isacharcolon}{\isacharcolon}\ {\isachardoublequoteopen}{\isacharparenleft}{\isacharprime}oid\ {\isasymtimes}\ {\isacharparenleft}{\isacharprime}oid{\isacharcomma}\ {\isacharprime}val{\isacharparenright}\ list{\isacharunderscore}op{\isacharparenright}\ list\ {\isasymRightarrow}\ {\isacharparenleft}{\isacharprime}oid\ list\ {\isasymtimes}\ {\isacharparenleft}{\isacharprime}oid\ {\isasymrightharpoonup}\ {\isacharprime}val{\isacharparenright}{\isacharparenright}{\isachardoublequoteclose}\isanewline
\isakeyword{where}\isanewline
\ \ {\isachardoublequoteopen}interp{\isacharunderscore}ops\ ops\ {\isasymequiv}\ foldl\ interp{\isacharunderscore}op\ {\isacharparenleft}{\isacharbrackleft}{\isacharbrackright}{\isacharcomma}\ Map{\isachardot}empty{\isacharparenright}\ ops{\isachardoublequoteclose}%
\end{isabelle}

\noindent \isa{list-order ops x y} holds iff, after interpreting the list of operations \isa{ops}, the list element with ID \isa{x} appears before the list element with ID \isa{y} in the resulting list.

\begin{isabelle}
\isacommand{definition}\isamarkupfalse%
\ list{\isacharunderscore}order\ {\isacharcolon}{\isacharcolon}\ {\isachardoublequoteopen}{\isacharparenleft}{\isacharprime}oid\ {\isasymtimes}\ {\isacharparenleft}{\isacharprime}oid{\isacharcomma}\ {\isacharprime}val{\isacharparenright}\ list{\isacharunderscore}op{\isacharparenright}\ list\ {\isasymRightarrow}\ {\isacharprime}oid\ {\isasymRightarrow}\ {\isacharprime}oid\ {\isasymRightarrow}\ bool{\isachardoublequoteclose}\ \isakeyword{where}\isanewline
\ \ {\isachardoublequoteopen}list{\isacharunderscore}order\ ops\ x\ y\ {\isasymequiv}\ {\isasymexists}xs\ ys\ zs{\isachardot}\ fst\ {\isacharparenleft}interp{\isacharunderscore}ops\ ops{\isacharparenright}\ {\isacharequal}\ xs\ {\isacharat}\ {\isacharbrackleft}x{\isacharbrackright}\ {\isacharat}\ ys\ {\isacharat}\ {\isacharbrackleft}y{\isacharbrackright}\ {\isacharat}\ zs{\isachardoublequoteclose}%
\end{isabelle}

\noindent The \isa{make-insert} function generates a new operation for insertion into a given index in a given list.
The exclamation mark is Isabelle's list subscript operator.

\begin{isabelle}
\isacommand{fun}\isamarkupfalse%
\ make{\isacharunderscore}insert\ {\isacharcolon}{\isacharcolon}\ {\isachardoublequoteopen}{\isacharprime}oid\ list\ {\isasymRightarrow}\ {\isacharprime}val\ {\isasymRightarrow}\ nat\ {\isasymRightarrow}\ {\isacharparenleft}{\isacharprime}oid{\isacharcomma}\ {\isacharprime}val{\isacharparenright}\ list{\isacharunderscore}op{\isachardoublequoteclose}\ \isakeyword{where}\isanewline
\ \ {\isachardoublequoteopen}make{\isacharunderscore}insert\ list\ val\ {\isadigit{0}}\ \ \ \ \ \ \ {\isacharequal}\ Insert\ None\ val{\isachardoublequoteclose}\ {\isacharbar}\isanewline
\ \ {\isachardoublequoteopen}make{\isacharunderscore}insert\ {\isacharbrackleft}{\isacharbrackright}\ \ \ val\ k\ \ \ \ \ \ \ {\isacharequal}\ Insert\ None\ val{\isachardoublequoteclose}\ {\isacharbar}\isanewline
\ \ {\isachardoublequoteopen}make{\isacharunderscore}insert\ list\ val\ {\isacharparenleft}Suc\ k{\isacharparenright}\ {\isacharequal}\ Insert\ {\isacharparenleft}Some\ {\isacharparenleft}list\ {\isacharbang}\ {\isacharparenleft}min\ k\ {\isacharparenleft}length\ list\ {\isacharminus}\ {\isadigit{1}}{\isacharparenright}{\isacharparenright}{\isacharparenright}{\isacharparenright}\ val{\isachardoublequoteclose}%
\end{isabelle}

\noindent The \isa{list-ops} predicate is a specialisation of \isa{spec-ops} to the \isa{list-op} datatype: it describes a list of (ID, operation) pairs that is sorted by ID, and can thus be used for the sequential interpretation of the OpSet.

\begin{isabelle}
\isacommand{fun}\isamarkupfalse%
\ list{\isacharunderscore}op{\isacharunderscore}deps\ {\isacharcolon}{\isacharcolon}\ {\isachardoublequoteopen}{\isacharparenleft}{\isacharprime}oid{\isacharcomma}\ {\isacharprime}val{\isacharparenright}\ list{\isacharunderscore}op\ {\isasymRightarrow}\ {\isacharprime}oid\ set{\isachardoublequoteclose}\ \isakeyword{where}\isanewline
\ \ {\isachardoublequoteopen}list{\isacharunderscore}op{\isacharunderscore}deps\ {\isacharparenleft}Insert\ {\isacharparenleft}Some\ ref{\isacharparenright}\ {\isacharunderscore}{\isacharparenright}\ {\isacharequal}\ {\isacharbraceleft}ref{\isacharbraceright}{\isachardoublequoteclose}\ {\isacharbar}\isanewline
\ \ {\isachardoublequoteopen}list{\isacharunderscore}op{\isacharunderscore}deps\ {\isacharparenleft}Insert\ \ None\ \ \ \ \ \ {\isacharunderscore}{\isacharparenright}\ {\isacharequal}\ {\isacharbraceleft}{\isacharbraceright}{\isachardoublequoteclose}\ \ \ \ {\isacharbar}\isanewline
\ \ {\isachardoublequoteopen}list{\isacharunderscore}op{\isacharunderscore}deps\ {\isacharparenleft}Delete\ \ ref\ \ \ \ \ \ \ \ {\isacharparenright}\ {\isacharequal}\ {\isacharbraceleft}ref{\isacharbraceright}{\isachardoublequoteclose}\isanewline
\isanewline
\isacommand{locale}\isamarkupfalse%
\ list{\isacharunderscore}opset\ {\isacharequal}\ opset\ opset\ list{\isacharunderscore}op{\isacharunderscore}deps\isanewline
\ \ \isakeyword{for}\ opset\ {\isacharcolon}{\isacharcolon}\ {\isachardoublequoteopen}{\isacharparenleft}{\isacharprime}oid{\isacharcolon}{\isacharcolon}{\isacharbraceleft}linorder{\isacharbraceright}\ {\isasymtimes}\ {\isacharparenleft}{\isacharprime}oid{\isacharcomma}\ {\isacharprime}val{\isacharparenright}\ list{\isacharunderscore}op{\isacharparenright}\ set{\isachardoublequoteclose}\isanewline
\isanewline
\isacommand{definition}\isamarkupfalse%
\ list{\isacharunderscore}ops\ {\isacharcolon}{\isacharcolon}\ {\isachardoublequoteopen}{\isacharparenleft}{\isacharprime}oid{\isacharcolon}{\isacharcolon}{\isacharbraceleft}linorder{\isacharbraceright}\ {\isasymtimes}\ {\isacharparenleft}{\isacharprime}oid{\isacharcomma}\ {\isacharprime}val{\isacharparenright}\ list{\isacharunderscore}op{\isacharparenright}\ list\ {\isasymRightarrow}\ bool{\isachardoublequoteclose}\ \isakeyword{where}\isanewline
\ \ {\isachardoublequoteopen}list{\isacharunderscore}ops\ ops\ {\isasymequiv}\ spec{\isacharunderscore}ops\ ops\ list{\isacharunderscore}op{\isacharunderscore}deps{\isachardoublequoteclose}%
\end{isabelle}

\subsubsection{Satisfying all conditions of $\mathcal{A}_\textsf{strong}$}

Part 1(a) of Attiya et al.'s specification states that whenever the list is observed, the elements of the list are exactly those that have been inserted but not deleted.
$\mathcal{A}_\textsf{strong}$ uses the visibility relation $\le_\textsf{vis}$ to capture the operations known to a node at some arbitrary point in the execution; in the OpSet model, we can simply prove the theorem for an arbitrary OpSet, since the contents of the OpSet at a particular time on a particular node correspond exactly to the set of operations known to that node at that time.

\begin{isabelle}
\isacommand{theorem}\isamarkupfalse%
\ inserted{\isacharunderscore}but{\isacharunderscore}not{\isacharunderscore}deleted{\isacharcolon}\isanewline
\ \ \isakeyword{assumes}\ {\isachardoublequoteopen}list{\isacharunderscore}ops\ ops{\isachardoublequoteclose}\isanewline
\ \ \ \ \isakeyword{and}\ {\isachardoublequoteopen}interp{\isacharunderscore}ops\ ops\ {\isacharequal}\ {\isacharparenleft}list{\isacharcomma}\ vals{\isacharparenright}{\isachardoublequoteclose}\isanewline
\ \ \isakeyword{shows}\ {\isachardoublequoteopen}a\ {\isasymin}\ dom\ {\isacharparenleft}vals{\isacharparenright}\ {\isasymlongleftrightarrow}\ {\isacharparenleft}{\isasymexists}ref\ val{\isachardot}\ {\isacharparenleft}a{\isacharcomma}\ Insert\ ref\ val{\isacharparenright}\ {\isasymin}\ set\ ops{\isacharparenright}\ {\isasymand}\isanewline
\hspace{47mm}{\isacharparenleft}{\isasymnexists}i{\isachardot}\ {\isacharparenleft}i{\isacharcomma}\ Delete\ a{\isacharparenright}\ {\isasymin}\ set\ ops{\isacharparenright}{\isachardoublequoteclose}
\end{isabelle}

Part 1(b) states that whenever the list is observed, the order of list elements is consistent with the global list order.
We can define the global list order simply as the list order that arises from interpreting the OpSet containing all operations in the entire execution.
Then, at any point in the execution, the OpSet is some subset of the set of all operations.

We can then rephrase condition 1(b) as follows: whenever list element \isa{x} appears before list element \isa{y} in the interpretation of \isa{some-ops}, then for any OpSet \isa{all-ops} that is a superset of \isa{some-ops}, \isa{x} must also appear before \isa{y} in the interpretation of \isa{all-ops}.
In other words, adding more operations to the OpSet does not change the relative order of any existing list elements.

\begin{isabelle}
\isacommand{theorem}\isamarkupfalse%
\ list{\isacharunderscore}order{\isacharunderscore}consistent{\isacharcolon}\isanewline
\ \ \isakeyword{assumes}\ {\isachardoublequoteopen}list{\isacharunderscore}ops\ some{\isacharunderscore}ops{\isachardoublequoteclose}\ \isakeyword{and}\ {\isachardoublequoteopen}list{\isacharunderscore}ops\ all{\isacharunderscore}ops{\isachardoublequoteclose}\isanewline
\ \ \ \ \isakeyword{and}\ {\isachardoublequoteopen}set\ some{\isacharunderscore}ops\ {\isasymsubseteq}\ set\ all{\isacharunderscore}ops{\isachardoublequoteclose}\isanewline
\ \ \ \ \isakeyword{and}\ {\isachardoublequoteopen}list{\isacharunderscore}order\ some{\isacharunderscore}ops\ x\ y{\isachardoublequoteclose}\isanewline
\ \ \isakeyword{shows}\ {\isachardoublequoteopen}list{\isacharunderscore}order\ all{\isacharunderscore}ops\ x\ y{\isachardoublequoteclose}
\end{isabelle}

Part 1(c) states that inserted elements appear at the specified position: that is, immediately after an insertion of \isa{oid} at index \isa{k}, the list index \isa{k} does indeed contain \isa{oid} (provided that \isa{k} is less than the length of the list).
We prove this property below.

\begin{isabelle}
\isacommand{theorem}\isamarkupfalse%
\ correct{\isacharunderscore}position{\isacharunderscore}insert{\isacharcolon}\isanewline
\ \ \isakeyword{assumes}\ {\isachardoublequoteopen}list{\isacharunderscore}ops\ {\isacharparenleft}ops\ {\isacharat}\ {\isacharbrackleft}{\isacharparenleft}oid{\isacharcomma}\ ins{\isacharparenright}{\isacharbrackright}{\isacharparenright}{\isachardoublequoteclose}\isanewline
\ \ \ \ \isakeyword{and}\ {\isachardoublequoteopen}ins\ {\isacharequal}\ make{\isacharunderscore}insert\ {\isacharparenleft}fst\ {\isacharparenleft}interp{\isacharunderscore}ops\ ops{\isacharparenright}{\isacharparenright}\ val\ k{\isachardoublequoteclose}\isanewline
\ \ \ \ \isakeyword{and}\ {\isachardoublequoteopen}list\ {\isacharequal}\ fst\ {\isacharparenleft}interp{\isacharunderscore}ops\ {\isacharparenleft}ops\ {\isacharat}\ {\isacharbrackleft}{\isacharparenleft}oid{\isacharcomma}\ ins{\isacharparenright}{\isacharbrackright}{\isacharparenright}{\isacharparenright}{\isachardoublequoteclose}\isanewline
\ \ \isakeyword{shows}\ {\isachardoublequoteopen}list\ {\isacharbang}\ {\isacharparenleft}min\ k\ {\isacharparenleft}length\ list\ {\isacharminus}\ {\isadigit{1}}{\isacharparenright}{\isacharparenright}\ {\isacharequal}\ oid{\isachardoublequoteclose}
\end{isabelle}

Part 2 states that the list order relation must be transitive, irreflexive, and total.
These three properties are straightforward to prove, using our definition of the \isa{list-order} predicate.

\begin{isabelle}
\isacommand{theorem}\isamarkupfalse%
\ list{\isacharunderscore}order{\isacharunderscore}trans{\isacharcolon}\isanewline
\ \ \isakeyword{assumes}\ {\isachardoublequoteopen}list{\isacharunderscore}ops\ ops{\isachardoublequoteclose}\isanewline
\ \ \ \ \isakeyword{and}\ {\isachardoublequoteopen}list{\isacharunderscore}order\ ops\ x\ y{\isachardoublequoteclose}\isanewline
\ \ \ \ \isakeyword{and}\ {\isachardoublequoteopen}list{\isacharunderscore}order\ ops\ y\ z{\isachardoublequoteclose}\isanewline
\ \ \isakeyword{shows}\ {\isachardoublequoteopen}list{\isacharunderscore}order\ ops\ x\ z{\isachardoublequoteclose}\isanewline
\isanewline
\isacommand{theorem}\isamarkupfalse%
\ list{\isacharunderscore}order{\isacharunderscore}irrefl{\isacharcolon}\isanewline
\ \ \isakeyword{assumes}\ {\isachardoublequoteopen}list{\isacharunderscore}ops\ ops{\isachardoublequoteclose}\isanewline
\ \ \isakeyword{shows}\ {\isachardoublequoteopen}{\isasymnot}\ list{\isacharunderscore}order\ ops\ x\ x{\isachardoublequoteclose}\isanewline
\isanewline
\isacommand{theorem}\isamarkupfalse%
\ list{\isacharunderscore}order{\isacharunderscore}total{\isacharcolon}\isanewline
\ \ \isakeyword{assumes}\ {\isachardoublequoteopen}list{\isacharunderscore}ops\ ops{\isachardoublequoteclose}\isanewline
\ \ \ \ \isakeyword{and}\ {\isachardoublequoteopen}x\ {\isasymin}\ set\ {\isacharparenleft}fst\ {\isacharparenleft}interp{\isacharunderscore}ops\ ops{\isacharparenright}{\isacharparenright}{\isachardoublequoteclose}\isanewline
\ \ \ \ \isakeyword{and}\ {\isachardoublequoteopen}y\ {\isasymin}\ set\ {\isacharparenleft}fst\ {\isacharparenleft}interp{\isacharunderscore}ops\ ops{\isacharparenright}{\isacharparenright}{\isachardoublequoteclose}\isanewline
\ \ \ \ \isakeyword{and}\ {\isachardoublequoteopen}x\ {\isasymnoteq}\ y{\isachardoublequoteclose}\isanewline
\ \ \isakeyword{shows}\ {\isachardoublequoteopen}list{\isacharunderscore}order\ ops\ x\ y\ {\isasymor}\ list{\isacharunderscore}order\ ops\ y\ x{\isachardoublequoteclose}
\end{isabelle}

\fi
\end{document}